\documentclass[sigconf]{acmart}

\usepackage{balance}  
\usepackage{epsfig}
\usepackage{epstopdf}
\usepackage{graphicx}
\usepackage{bookmark}
\usepackage{url}
\usepackage[ruled,boxed,vlined,linesnumbered]{algorithm2e}
\usepackage{mathrsfs}
\usepackage{graphics}
\usepackage{subcaption}
\usepackage{color}
\usepackage{booktabs}
\usepackage{amssymb}
\usepackage{footnote}
\usepackage{wrapfig}
\makesavenoteenv{tabular}
\makesavenoteenv{table}

\newtheorem{defn}{Definition}
\newtheorem{thm}{Theorem}

\newcommand{\tabincell}[2]{\begin{tabular}{@{}#1@{}}#2\end{tabular}}
\copyrightyear{2018} 
\acmYear{2018} 
\setcopyright{acmlicensed}
\acmConference[SIGIR '18]{The 41st International ACM SIGIR Conference on Research \& Development in Information Retrieval}{July      8--12, 2018}{Ann Arbor, MI, USA}
\acmBooktitle{SIGIR '18: The 41st International ACM SIGIR Conference on Research \& Development in Information Retrieval, July      8--12, 2018, Ann Arbor, MI, USA}
\acmPrice{15.00}
\acmDOI{10.1145/3209978.3210038}
\acmISBN{978-1-4503-5657-2/18/07}

\acmConference[SIGIR'18]{ACM SIGIR conference}{July 2018}{Ann Arbor, MI, USA}

\begin{document}
\title{Top-k Route Search through Submodularity Modeling of Recurrent POI Features}


\author{Hongwei Liang}
\affiliation{%
  \institution{School of Computing Science \\ Simon Fraser University, Canada}
}
\email{hongweil@sfu.ca}

\author{Ke Wang}
\affiliation{%
  \institution{School of Computing Science \\ Simon Fraser University, Canada}
}
\email{wangk@cs.sfu.ca}


\begin{abstract}
We consider a practical top-$k$ route search problem: given a collection of points of interest (POIs) with rated features and traveling costs between POIs, a user wants to find $k$ routes from a source to a destination and limited in a cost budget, that maximally match her needs on feature preferences. 
One challenge is 
dealing with the personalized diversity requirement where users have various trade-off between 
quantity (the number of POIs with a specified feature) and variety (the coverage of specified features). 
Another challenge is the large scale of the POI map and the great many
alternative routes to search. We model 
the personalized diversity requirement by the whole class of submodular functions, and
present an optimal solution to the top-$k$ route search problem through indices
for retrieving relevant POIs in both feature and route spaces 
and various strategies for pruning the search space using user preferences and constraints. 
We also present promising heuristic solutions and evaluate all the solutions on real life data.
\end{abstract}

%
%

%

\keywords{Location-based Search; Route Planning; Diversity Requirement}

\maketitle

\section{Introduction} \label{sec:intro}

%

The
dramatic growth of publicly accessible mobile/geo-tagged data has triggered a revolution in location based services \cite{junglas2008location}.
An emerging thread is route planning, with pervasive
applications in trip recommendation, intelligent navigation, ride-sharing, and augmented-reality gaming, etc.
According to \cite{wttc2017}, the travel and tourism industry directly and indirectly contributed US\$7.6 trillion to the global economy and supported 292 million jobs in 2016.
The majority of current route planning systems yields shortest paths or explores popular POIs \cite{zheng2009mining}, or recommends routes based on users' historical records \cite{kurashima2010travel} or crowdsourced experience \cite{quercia2014shortest}.

A practical problem that has not been well studied is that, a user wants to be suggested a small number of routes that not only satisfy her cost budget and spatial constraints, but also best meet
her personalized requirements on diverse POI features.
We instantiate this problem with a travel scenario.
Consider that a new visitor to Rome
wishes to be recommended a trip, starting from her hotel and ending at the airport, that allows her to visit museums, souvenir shops, and eat at some good Italian restaurants (not necessarily in this order) in the remaining hours before taking the flight.
She values the variety over the number of places visited, e.g., 
a route consisting of one museum, one shop, and one Italian restaurant is
preferred to a route consisting of two museums and two shops.


\begin{figure}[t]%
        \centering
        \includegraphics[width=2.7in]{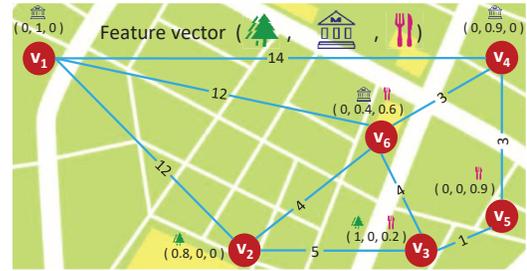}
        \caption{A sample POI map. Each node $v_i$ represents a POI with 3 features (Park, Museum, Restaurant). Each feature has a numeric rating in the range $[0,1]$, indicated by the vector aside the POI. Each edge has an associated traveling cost.
        }
        \label{fig:poi_map}
\end{figure}

%

The above problem is actually generalizable to various route planning 
scenarios, and they illustrate some common structures and requirements.
First, there is a \emph{POI map} where POIs are connected by edges with traveling cost between POIs, and each POI has a location and 
is associated with a vector of features (e.g., museum) with numeric or binary ratings. The POI map can be created from
Google Map, and features and ratings of POIs can be created from user rating and text tips 
available on location-based services such as Foursquare, or 
extracted from check-ins and
user provided reviews \cite{diao2014jointly}.
Second, the user seeks to find \emph{top-$k$ routes} $\{\mathcal{P}_1,\cdots,\mathcal{P}_k\}$, from a specified source $x$ to a specified destination $y$ within a travel cost budget $b$, that have highest values of a certain gain function 
$Gain({\mathcal{P}_i}_{V})$ for the set of POIs ${\mathcal{P}_i}_{V}$ on the routes $\mathcal{P}_i$. The 
user specifies her preference of routes through a weight vector $\mathbf{w}$ with $\mathbf{w}_{h}$ being the weight of a 
feature $h$, and 
a \emph{route diversity requirement}, which specifies 
a trade-off between quantity (the number of POIs with a preferred feature) and variety (the coverage of preferred features)
for the POIs on a route. 
The gain function 
has the form
$Gain(\mathcal{P}_V) = {\sum}_{h} \mathbf{w}_{h} \Phi_h(\mathcal{P}_V)$, where 
$\Phi_h$ for each feature $h$ aggregates the feature's scores of the POIs $\mathcal{P}_V$.


%

To better motivate the route diversity requirement, let us 
consider the POI map in Figure \ref{fig:poi_map} and a user with the source $v_1$, destination 
$v_5$ and the budget $b = 18$. The user weights the features Park and Museum using the vector  $\mathbf{w} = (0.4, 0.6, 0)$, and values \emph{both} quantity and variety.
If the sum aggregation $\Phi_h$ is used, the route $v_1 \rightarrow v_6 \rightarrow v_4 \rightarrow v_5$ will have 
the highest $Gain$. However, the user may not prefer this route because it does not include any park though
it includes 3 museums. With the max aggregation used, 
the route $v_1 \rightarrow v_3 \rightarrow v_5$ has the highest $Gain$ by including one top scored museum and one top scored park, but this route does not maximally use the entire budget available. 
Intuitively, the sum aggregation is ``quantity minded'' but ignores variety, whereas 
the max aggregation is the opposite; neither models a 
proper trade-off between quantity and variety as the user considered. 
The above user more prefers the route $v_1 \rightarrow v_2 \rightarrow v_3 \rightarrow v_5$ that visits 
multiple highly scored museums and parks, which will better address both quantity and variety.

Solving the top-$k$ route search problem
faces two challenges. 

\textbf{Challenge I}. One challenge is to design a general enough $\Phi_h$ that includes a large class of aggregation
functions to model a \emph{personalized} route diversity requirement where each user 
has her own quantity and variety trade-off. Our approach is treating 
the satisfaction by visiting each POI as the marginal utility
and modeling the aggregation of such utilities of POIs with 
the diminishing marginal utility property by \emph{submodular} set functions $\Phi_h$.
The intuition is that, as the user visits more POIs of the same, 
her marginal gain from such visits decreases gradually. 
Submodularity has been used for modeling user behaviors in many real world problems 
\cite{krause2008beyond}\cite{kempe2003maximizing}.
To the best of our knowledge, modeling user's diversity requirement on a route
by submodularity has not been done previously.


%




\textbf{Challenge II}. The top-$k$ route problem is NP-hard as it subsumes the NP-hard orienteering problem \cite{chekuri2005recursive}.
However, users typically demand the routes not only be in high quality, even optimal, but also be generated in a timely manner (seconds to minutes).
Fortunately, the users' preferences and constraints on desired routes provide new opportunities 
to reduce the search space and find optimal top-$k$ routes with fast responses.
For example, for a user with only 6-hour time budget and preferring museums and parks on a
route, all the POIs in other types or beyond the 6 hours limit will be irrelevant.
The key of an exact algorithm is to prune, as early as possible, such irrelevant POIs as well as the routes that are unpromising to make into the top-$k$ list due to a low gain $Gain(\mathcal{P}_V)$.
However, this task is complicated by the incorporation of a generic submodular aggregation function $\Phi_h$
motivated above in our objective $Gain(\mathcal{P}_V)$, and designing a tight \emph{upper bounding} strategy on $Gain(\mathcal{P}_V)$ 
for pruning unpromising routes is a major challenge.


\textbf{Contributions}.
The main contributions of this paper are:

$\bullet$
We define the top-$k$ route search problem with a new personalized route diversity 
requirement where the user can choose any submodular function  $\Phi_h$ to model her desired level of diminishing return. 
As an instantiation, we show that the family of power-law functions is a sub-family of submodular functions 
and can model a spectrum of personalized diversity requirement. 
(Section \ref{sec:prob})

$\bullet$
Our first step towards an efficient solution is to eliminate irrelevant POIs for a query, by proposing 
a novel structure for indexing the POI map
on both features and travel costs. 
This index reduces the POI map to a small set of POIs for a query.(Section \ref{sec:overview})

$\bullet$
Our second step towards an efficient solution is to
prune unpromising routes, by proposing a novel optimal algorithm, \texttt{PACER}. 
The novelties of the algorithm include an elegant route enumeration strategy for a compact representation
of search space and the reuse of computed results, a cost-based pruning for eliminating non-optimal routes, 
and a gain-based upper bound strategy for pruning routes that cannot make into the top-$k$ list. 
The algorithm works for \emph{any} submodular function $\Phi_h$. (Section \ref{sec:top-k})

$\bullet$
To deal with the looser query constraints, 
we present two heuristic algorithms with a good efficiency-accuracy trade-off, by finding a good solution with 
far smaller search spaces. (Section \ref{sec:heu})

$\bullet$
We evaluate our algorithms on real-world data sets. 
\texttt{PACER} provides optimal answers while being orders of magnitude faster than the baselines. 
The heuristic algorithms 
provide answers of a competitive quality and work efficiently for a larger POI map and/or a
looser query constraint. 
(Section \ref{se:exp})



\section{Related Work} \label{sec:relate}

Route recommendation/planning that suggests a POI sequence or a path is related to our work.
Works like \cite{kurashima2010travel} \cite{de2010automatic} learn from historical travel behaviors and recommend routes by either sequentially predicting the next location via a Markov model or globally constructing a route.
These works rely on users' historical visit data, thus, 
cannot be applied to a new user with no visit data or a user with dynamically changed preferences. 
\cite{basu2011interactive} interactively plans a route based on user feedback at each step. 
Our approach does not rely on user's previous visit data or interactive feedback, and works for any users by modeling the preferences through a query.
%

Several works recommend a route by maximizing user satisfaction under certain constraints. 
\cite{gionis2014customized} assumes that each POI has a single type and searches for a route with POIs following a pre-determined order of types.
\cite{zhang2016trip} allows the user to specify a minimum number of POI types, instead of exact types, in a route.
\cite{lu2012personalized} estimates temporal-based user preferences.
\cite{lim2017personalized} focuses on modeling the queuing time on POIs.
\cite{cao2012keyword} constructs an optimal route covering user-specified
categories associated with locations. None of them considers a general 
route diversity requirement for modeling user's quantity and variety trade-off. 

\cite{zeng2015optimal}, perhaps most related to our work, adopts a
keyword coverage function to measure the degree to which query keywords are covered by a route, similar to ours.
Their pruning strategies are designed specifically for their specific keyword coverage function; thus, does not address the personalized route diversity requirement, where a different submodular function may be required.
Our pruning strategies apply to any submodular function $\Phi_h$. 
Finally \cite{zeng2015optimal} produces a single route, and its performance 
is only ``2-3 times faster than the brute-force algorithm", as pointed in \cite{zeng2015optimal}. 

Less related to our work is the next POI recommendation \cite{zhang2015location} that aims to recommend the POI to be visited next, and the travel package recommendation \cite{liu2011personalized} that aims to recommend a set of POIs.  They are quite different from our goal of finding a route as a sequence of featured POIs. 
Trajectory search 
either retrieves \emph{existing} (segments of) trajectories that match certain similarity query  \cite{zheng2015approximate} from a database, or constructs a route based on the retrieved trajectories \cite{dai2015personalized}.
These works assume the existence of a trajectory database, instead of a POI map for route construction.

The classical Orienteering Problem (OP), such as  \cite{chekuri2005recursive}, studied in operational research on theoretical level, finds a path, limited in length, that visits some nodes and maximizes a global reward collected from the nodes on the path.
No POI feature or route diversity requirement is considered in OP.
\section{Preliminary} \label{sec:prob}

\begin{table}[]   %
\centering
\caption{Nomenclature}
\label{tab:notation}
\small
\begin{tabular}{cl} \toprule
\multicolumn{1}{c}{\textbf{Notation}} & \multicolumn{1}{c}{\textbf{Interpretation}}  \\ \toprule
	$\mathbf{F} \in \mathbb{R}^{|\mathcal{V}| \times |\mathcal{H}|}$ & POI-feature matrix $\mathbf{F}$ with POI set $\mathcal{V}$ and feature set $\mathcal{H}$ \\		
	$\mathbf{F}_{i,h}$ & the rating on feature $h \in \mathcal{H}$ for POI $i \in \mathcal{V}$ \\
	$s_i$ & staying cost on POI $i$ \\			
	$t_{i,j}$ & the traveling cost on edge $e_{i,j} \in \mathcal{E}$ \\		
	$T_{i,j}$ & the least traveling cost from any POI $i$ to any POI $j$ \\	
	$\mathcal{P}$, $\mathcal{P}_V$ & route $\mathcal{P}$ with the included POI set $\mathcal{P}_V$ \\ \midrule
	\tabincell{c}{$Q = $ \\$(x, y, b,$ \\ $\mathbf{w}, \boldsymbol{\theta}, \Phi$)} & \tabincell{l}{user query with parameters:\\ $x$ and $y$ -- source and destination location \\$b$ -- travel cost budget \\$\mathbf{w} \in \mathbb{R}^{|\mathcal{H}|}$ -- feature preference vector \\$\boldsymbol{\theta} \in \mathbb{R}^{|\mathcal{H}|}$ -- filtering vector on feature ratings \\ $\Phi$ -- feature aggregation functions}	\\ \midrule
	$\mathcal{V}_Q$, $n$ & POI candidates set $\mathcal{V}_Q$ retrieved by $Q$ with its size $n$ \\ 
	$\tilde{\mathbf{F}}_{i,h}$ & $\mathbf{F}_{i,h}$ after filtered by $\boldsymbol{\theta}$ \\ 
	$Gain(\mathcal{P}_V, Q)$ & \emph{gain} of a route $\mathcal{P}$ given query $Q$ \\ 
	\bottomrule

\end{tabular}
\vspace{-3pt}
\end{table}

Table \ref{tab:notation} summarizes the notations frequently
used throughout the paper.
The variables in \textbf{bold-face} are vectors or matrices.

\subsection{Problem Statement} \label{sec:interest}

\begin{defn}\label{def:map}
\textbf{[A POI Map]} A POI map $G=(\mathcal{V},\mathcal{E})$ is a directed/undirected and connected graph, where $\mathcal{V}$ is a set of geo-tagged POI nodes and $\mathcal{E} \subseteq \mathcal{V} \times \mathcal{V}$ is a set of edges between nodes $(i,j)$, $i,j \in \mathcal{V}$.
$\mathcal{H}$ is a set of features on POIs.
$\mathbf{F} \in \mathbb{R}^{|\mathcal{V}| \times |\mathcal{H}|}$ denotes the POI-feature matrix, where $\mathbf{F}_{i,h} \in [0,\beta]$ is the rating on a feature $h$ for the POI $i$.
Each POI $i \in \mathcal{V} $ is associated with a staying cost $s_i$.
Each edge $e_{i,j} \in \mathcal{E}$ has a travel cost $t_{i,j}$. $\Box$
\end{defn}

The choices of $s_i$ and $t_{i,j}$ can be time, expenses, or other costs.

\begin{defn}
\label{def:trip}
\textbf{[Routes]} A route $\mathcal{P}$ is a path $x \rightarrow \cdots i \cdots \rightarrow y$ in $G$ from the origin $x$ to the destination $y$ through a sequence of non-repeating POIs $i$ except possibly $x = y$.
$\mathcal{P}_V$ denotes the set of POIs on $\mathcal{P}$.
$T_{i,j}$ denotes the least traveling cost from $i$ to the next visited $j$, where $i,j$ are not necessarily adjacent in $G$.
The cost of $\mathcal{P}$ is 
\begin{equation}
cost(\mathcal{P})={\sum}_{i \in \mathcal{P}_V} s_i + {\sum}_{i \rightarrow j \in \mathcal{P} } \ T_{i,j}. \ \ \ \ \ \ \Box\
\end{equation}
\end{defn}

A route $\mathcal{P}$ includes only the POIs $i$ that the user actually ``visits'' by staying at $i$ with $s_i > 0$.
Each $i \rightarrow j$ on a route is a path from $i$ to $j$ with the least traveling cost $T_{i,j}$. 
The intermediate POIs between $i,j$ on path $i \rightarrow j$ are not included in $\mathcal{P}$.
The staying times at $x$ and/or $y$ can be either considered
or ignored depending on the user choice. The latter case can
be modeled by setting $s_x = s_y = 0$.


At the minimum, the user has an origin $x$ and a destination $y$ for a route, not necessarily distinct, and a budget $b$ on the cost of the route.
In addition, the user may want the POIs to have certain features specified by a $|\mathcal{H}|$-dimensional weight vector $\mathbf{w}$ with each element $\mathbf{w}_{h} \in [0,1]$ and $\Sigma_h \mathbf{w}_{h}=1$.
The user can also specify a filtering vector $\boldsymbol{\theta}$ so that $\mathbf{F}_{i,h}$ is set to 0 if it is less than $\boldsymbol{\theta}_{h}$.
$\tilde{\mathbf{F}}_{i,h}$ denotes $\mathbf{F}_{i, h}$ after this filtering.
Finally, the user may specify a route diversity requirement through a feature aggregation function vector $\Phi$, with $\Phi_{h}$ for each feature $h$. $\Phi_{h}(\mathcal{P}_V)$ aggregates
the rating on feature $h$ over the POIs in $\mathcal{P}_V$.
See more details in Section \ref{sec:ref_prob}.

\begin{defn}\label{def:user}
\textbf{[Query and Gain]} A query $Q$ is a 6-tuple $(x,y, b,\mathbf{w},$ $\boldsymbol{\theta}, \Phi)$.
A route $\mathcal{P}$ is \emph{valid} if $cost(\mathcal{P})\leq b$.
The \emph{gain} of $\mathcal{P}$ w.r.t. $Q$ is
\begin{equation} \label{eq:gain_power}
Gain(\mathcal{P}_V,Q) = {\sum}_{h} \mathbf{w}_{h}  \Phi_{h}(\mathcal{P}_V).
\ \ \ \ \ \ \Box
\end{equation}
\end{defn}

Note that only the specification of $x,y, b$ is required; if the specification of $\mathbf{w}, \boldsymbol{\theta}, \Phi$ is not provided by a user,
their default choices can be used, or can be learned from users' travel records if such data are available (not the focus of this paper).
$Gain(\mathcal{P}_V,Q)$ is a set function and all routes $\mathcal{P}$
that differ only in the order of POIs have the same $Gain$, and the order of POIs affects only $cost(\mathcal{P})$.


\textbf{[Top-$k$ route search problem]}
Given a query $Q$ and an integer $k$, the goal is to find  $k$ valid routes $\mathcal{P}$ that have different POI sets $\mathcal{P}_V$ (among all routes having the same $\mathcal{P}_V$,
we consider only the route with the smallest $cost(\mathcal{P})$) and the highest $Gain(\mathcal{P}_V, Q)$ (if ties, ranked by $cost(\mathcal{P})$). The $k$ routes are denoted by $topK$.
%

In the rest of the paper, we use $Gain(\mathcal{P}_V)$ for $Gain(\mathcal{P}_V,Q)$.



%

\subsection{Modeling Route Diversity Requirement}\label{sec:ref_prob}
To address the personalized route diversity requirement, we consider a submodular
$\Phi_h$ to model the diminishing marginal utility
as more POIs with feature $h$ are added to a route.
 A set function $f: 2^V \rightarrow \mathbb{R}$ is \emph{submodular} if for every $X \subseteq Y \subseteq V$ and $v \in V\setminus Y$, $f(X \cup \{v\}) - f(X) \geq f(Y \cup \{v\}) - f(Y)$, and is \emph{monotone}
if for every $X \subseteq Y \subseteq V$, $f(X) \leq f(Y)$.
The next theorem follows from \cite{krause2012submodular} and the fact that  $Gain(\mathcal{P}_V)$ is a nonnegative linear combination of $\Phi_h$.

\begin{thm}\label{thm:sub}
If for every feature $h$, $\Phi_{h}(\mathcal{P}_V)$ is nonnegative, monotone and submodular,
so is $Gain(\mathcal{P}_V)$.
\end{thm}

We aim to provide a general solution to the top-$k$ route search problem for any nonnegative, monotone and submodular $\Phi_h$,
which model various personalized route diversity requirement.
To illustrate the modeling power of such $\Phi_h$,
for example, consider $\Phi_h$ defined by the power law function
\begin{equation} \label{eq:gain_power_2}
\Phi_{h}(\mathcal{P}_V) = {\sum}_{i \in \mathcal{P}_V} {R_{h}(i)}^{-\boldsymbol{\alpha}_h} \tilde{\mathbf{F}}_{i,h},
\end{equation}
where $R_{h}(i)$ is the rank of POI $i$ on the rating of feature $h$ among all the POIs in $\mathcal{P}_V$ (the largest value ranks the first), rather than the order that $i$ is added to $\mathcal{P}$, and $\boldsymbol{\alpha}_h \in [0, +\infty)$ is the power law exponent for feature $h$.
${R_{h}(i)}^{-\boldsymbol{\alpha}_h}$ is non-increasing as $R_{h}(i)$ increases.
For a sample route $\mathcal{P} = A(3) \rightarrow B(5)$ with the ratings of feature $h$ for each POI in the brackets, and $\boldsymbol{\alpha}_h = 1$,
the ranks for $A$ and $B$ on $h$ are $R_{h}(A) = 2$ and $R_{h}(B) = 1$. Thus,
$\Phi_{h}(\mathcal{P}_V) = 2^{-1} \times 3 + 1^{-1} \times 5$, with a diminishing factor $2^{-1}$ for the secondly ranked $A$.
If we use a larger $\boldsymbol{\alpha}_h = 2$, $\Phi_{h}(\mathcal{P}_V) = 2^{-2} \times 3 + 1^{-2} \times 5$ has a larger diminishing factor for $A$.

In general, a larger $\boldsymbol{\alpha_h}$ means a faster diminishing factor for the ratings $\tilde{\mathbf{F}}_{i,h}$ on the recurrent feature, i.e., a diminishing marginal value on $h$.
Note that the sum aggregation ($\boldsymbol{\alpha_h} = 0$) and the max aggregation ($\boldsymbol{\alpha_h}=\infty$)
are the special cases. 
Hence, Eqn. (\ref{eq:gain_power_2}) supports a spectrum of diversity requirement 
through the setting of $\boldsymbol{\alpha_h}$.

Note that  $R_{h}(i)$ for an existing POI $i$ may decrease when a new POI $j$ is added to $\mathcal{P}$, so it is incorrect to compute the new $\Phi_{h}(\mathcal{P}_V)$ by simply adding the marginal brought by $j$ to existing value of $\Phi_{h}(\mathcal{P}_V)$. For ease of presentation, we assume $\boldsymbol{\alpha}_h$ has same value for all $h$
and use $\boldsymbol{\alpha}$ for $\boldsymbol{\alpha}_h$ in the rest of the paper.



\begin{thm} \label{thm:sub_agg}
$\Phi_{h}(\mathcal{P}_V)$ defined in 
Eqn. (\ref{eq:gain_power_2}) is nonnegative, monotone and submodular.
\end{thm}
\begin{proof}
The nonnegativity and monotonicity of $\Phi_{h}(\mathcal{P}_V)$ in Eqn. (\ref{eq:gain_power_2}) is straightforward. 
For its submodularity, we omit the mathematical proof due to limited space and only present an intuitive idea. 
Let $X$ and $Y$ be the set of POIs in two routes, $X \subseteq Y$.
Intuitively, for every $i \in X$, there must be $i \in Y$ and $i$'s rank in $X$ is not lower than that in $Y$.
Let $v \in \mathcal{V}\setminus Y$ so that $X' = X \cup \{v\}$ and $Y' = Y \cup \{v\}$.
Similarly, $v$' rank in $X'$ is not lower than that in $Y'$, thus the increment brought by $v$ to $X$ is not less than that to $Y$, which means $\Phi_{h}(X') - \Phi_{h}(X) \geq \Phi_{h}(Y') - \Phi_{h}(Y)$.
Hence, it is submodular.
\end{proof}

The user can also personalize her diversity requirement by specifying any other submodular $\Phi_h$, such as a log utility function
$\Phi_{h}(\mathcal{P}_V) = log(1 + {\sum}_{i \in \mathcal{P}_V} \tilde{\mathbf{F}}_{i,h})$ and the coverage function $\Phi_{h}(\mathcal{P}_V) = 1- {\prod}_{i \in \mathcal{P}_V} [1 - \tilde{\mathbf{F}}_{i,h}]$.
Our approach only depends on the submodularity of $\Phi_h$, but is independent of the exact choices of such functions.

Our problem subsumes two NP-hard problems, i.e., the submodular maximization problem \cite{krause2012submodular} and the orienteering problem \cite{chekuri2005recursive}.

\subsection{Framework Overview}
To efficiently deal with the high computational complexity of this problem, we divide the overall framework into
the offline component and the online component.
Before processing any query, the offline component carefully indexes the POI map on feature and cost dimensions for speeding up future POI selection and travel cost computation. 
The online component responds to the user query $Q$ with 
\emph{Sub-index Retrieval} that extracts the sub-indices relevant to $Q$, 
and \emph{Routes Search} that finds the top-$k$ routes using the sub-indices.
For routes search, as motivated in Section \ref{sec:intro}, we consider both the exact algorithm with novel pruning strategies,
and heuristic algorithms to deal with the worst case of less constrained $Q$.


We first consider an indexing strategy 
in Section \ref{sec:overview}, and then consider routes search algorithms in Section \ref{sec:top-k} and Section \ref{sec:heu}.

\section{Indexing} \label{sec:overview}

In this section, we explain the offline indexing component and the Sub-index Retrieval of the online component.

\subsection{Offline Index Building} \label{sec:indexing}
The POI map data is stored on disk. To answer user queries rapidly with low I/O access and speed up
travel cost computation,
we build two indices, $\mathsf{FI}$ and $\mathsf{HI}$ stored on disk.

$\mathsf{\textbf{FI}}$ is an inverted index mapping each feature $h$ to a list of POIs having non-zero rating on $h$.
An entry $(v_i, \mathbf{F}_{i,h})$ indicates the feature rating $\mathbf{F}_{i,h}$ for POI $v_i$, sorted in descending order of $\mathbf{F}_{i,h}$.
$\mathsf{\textbf{FI}}$ helps retrieving the POIs related to the features
specified by a query. 



The least traveling cost $T_{i,j}$ between two arbitrary POIs $i$ and $j$ is frequently required in the online component. To compute $T_{i,j}$ efficiently,
we employ the 2-hop labeling \cite{jiang2014hop} for point-to-point shortest distance querying on weighted graphs. 
\cite{jiang2014hop} shows scalable results for finding 2-hop labels for both unweighted and weighted graphs, and the constructed labels return \emph{exact} shortest distance queries.
Our $\mathsf{HI}$ index is built using the 2-hop labeling method. 


$\mathsf{\textbf{HI}}$.
For an undirected graph, there is one list of labels for each node $v_i$, where each label $(u, d)$ contains a node $u \in \mathcal{V}$, called \emph{pivot}, and the least traveling cost $d$ between $v_i$ and $u$. $\mathsf{HI}(v_i)$
denotes the list of labels for $v_i$, sorted in the ascending order of $d$.
According to \cite{jiang2014hop}, $T_{i,j}$
between $v_i$ and $v_j$
is computed by
%
%
\begin{equation} \label{eq:cost}
T_{i,j} = \underset{(u, d_1) \in  \mathsf{HI}(v_i) \cap (u, d_2) \in\mathsf{HI}(v_j) }{ \min } (d_1 + d_2).
\end{equation}

\begin{figure}[]%
        \centering
        \includegraphics[width=2.8in]{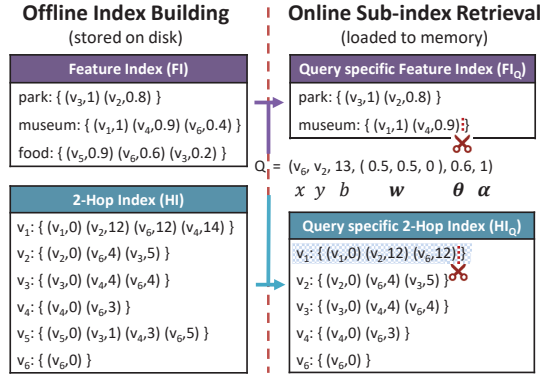}
        \caption{Left part: $\mathsf{FI}$ and $\mathsf{HI}$ built from the POI map in Figure \ref{fig:poi_map}.
        Right Part: Given a query $Q$, retrieve POI candidates $\mathcal{V}_{Q}$ by
         retrieving the subindices $\mathsf{FI_Q}$ and $\mathsf{HI_Q}$ from $\mathsf{FI}$ and $\mathsf{HI}$.
        }
        \label{fig:index}
\end{figure}

Figure \ref{fig:index} (left part) shows the FI and HI for the POI map in Figure \ref{fig:poi_map}.
For example, to compute $T_{2,5}$, we search for the common pivot nodes $u$ from the pivot label lists of $v_2$ and $v_5$ and find that $v_3$ minimizes the traveling cost between $v_2$ and $v_5$, so $T_{2,5} = 5 + 1 = 6$.

In the case of a directed graph, each POI $v_i$ will have
two lists of labels in $\mathsf{HI}$, $\mathsf{HI}(v_i^{out})$ for $v_i$ as the source,
and $\mathsf{HI}(v_i^{in})$ for $v_i$ as the destination. And we simply replace $v_i$ with $v_i^{out}$ and $v_j$ with $v_j^{in}$ in Eqn. (\ref{eq:cost}) to compute $T_{i,j}$.

\subsection{Online Sub-index Retrieval} \label{sec:filter}
Given a query $Q$, the first thing is to retrieve the POI candidates $\mathcal{V}_Q$ that are likely to be used in the routes search part.
In particular, the POIs that do not contain any feature in the preference vector $\mathbf{w}$ or do not pass any threshold in $\boldsymbol{\theta}$ will never be used, nor the ones that cannot be visited on the way from the source $x$ to the destination $y$ within the budget $b$.
This is implemented by retrieving the query specific sub-indices $\mathsf{FI_Q}$ from $\mathsf{FI}$ and $\mathsf{HI_Q}$ from $\mathsf{HI}$.

Figure \ref{fig:index} (right part) shows how the retrieval works for a query 
$Q = (x = v_6, y = v_2, b = 13, \mathbf{w} = (0.5, 0.5, 0), \boldsymbol{\theta} = 0.6, \boldsymbol{\alpha} = 1)$,
where the weights in $\mathbf{w}$ are for (Park, Museum, Food), and $\boldsymbol{\alpha}$ is the power law exponent in Eqn. (\ref{eq:gain_power_2}).
Here the elements in each vector $\boldsymbol{\theta}$ and $\boldsymbol{\alpha}$ have the same value for all features.

$\mathsf{\textbf{FI}_\textbf{Q}}$, a sub-index of $\mathsf{FI}$, is retrieved using $\mathbf{w}$ and $\boldsymbol{\theta}$.  $\mathbf{w}$ directly locates the lists for the user preferred (with $\mathbf{w}_{h} > 0$) features. $\boldsymbol{\theta}$ is used to cut off lower rated POIs on the sorted lists indicated by red scissors. $\mathcal{V}_Q = \{v_1, v_2, v_3, v_4\}$ contains the remaining POIs.

$\mathsf{\textbf{HI}_\textbf{Q}}$, a sub-index of $\mathsf{HI}$, is then formed by retrieving the lists for each POI in $\mathcal{V}_{Q}$ and also those for $x$ and $y$, and $b$ is used to cut off the sorted lists, indicated by red scissors.
We also check whether a POI $i$ in current $\mathcal{V}_{Q}$ is actually reachable by checking the single-point visit cost: if $s_x + T_{x,i} + s_i + T_{i, y} + s_y > b$, we remove $i$ from $\mathcal{V}_{Q}$ and remove its list from $\mathsf{HI_Q}$, as indicated by the blue shading.
Then we get the final POI candidates $\mathcal{V}_{Q}$.
Typically, $|\mathcal{V}_{Q}| \ll |\mathcal{V}|$.




$\mathsf{\textbf{FI}_\textbf{Q}}$ and $\mathsf{\textbf{HI}_\textbf{Q}}$ are retrieved only once and kept in memory.

\section{Optimal Routes Search} \label{sec:top-k}
With POI candidate set $\mathcal{V}_{Q}$ and the sub-indices extracted, the next step is the Routes Search phase.
We present an optimal routes search algorithm in this section.
Considering the complexity and generality of the problem, 
a standard tree search or a traditional algorithm for the orienteering problem does not work.
An ideal algorithm design should meet the following goals:
i. search all promising routes in a smart manner without any redundancy;
ii. prune unpromising routes as aggressively as possible while preserving the optimality of the top-$k$ answers;
iii. ensure that the search and pruning strategies are applicable to any nonnegative, monotone and submodular aggregation functions $\Phi_h$.
To this end, we propose a novel algorithm, \textbf{P}refix b\textbf{A}sed \textbf{C}ompact stat\textbf{E}s g\textbf{R}owth (\textbf{PACER}), that incorporates the idea of dynamic programming
and fuses a cost-based pruning strategy and a gain-based pruning strategy 
in an unified way. 






Next, we present our enumeration and
pruning strategies, followed by the detailed algorithm and the complexity analysis.

\subsection{Prefix-based Compact State Growth} \label{sec:enumeration}
A route $\mathcal{P}$ is associated with several variables:  $\mathcal{P}_{V}$, $Gain(\mathcal{P}_V)$, the ending POI $end(\mathcal{P})$, and $cost(\mathcal{P})$.
If $x$ is not visited, $s_x$ and $\tilde{\mathbf{F}}_{x,h}$ for every $h$ are set to 0;
the same is applied to $y$.
A POI sequence is an \emph{open route} if it starts from $x$ and visits several POIs other than $y$; it is a \emph{closed route} if it starts from $x$ and ends at $y$. The initial open route includes only $x$.
An open route $\mathcal{P}$ is \emph{feasible} if its closed form $\mathcal{P} \rightarrow y$  satisfies $cost(\mathcal{P} \rightarrow y)\leq b$. In the following discussion, $\mathcal{P}$ denotes either an open route or a closed route.
An open route $\mathcal{P}^{-}$ with $end(\mathcal{P}^-)=i$ can be extended into a longer open route $\mathcal{P} = \mathcal{P}^{-} \rightarrow j$ by a POI $j \not\in \mathcal{P}^{-}_{V} \cup \{y\}$. The variables for  $\mathcal{P}$ are
\begin{equation}\label{eq:new_state}
\left\{
  \begin{array}{l}
    \mathcal{P}_V = \mathcal{P}^{-}_V \cup \{j\} \\
    Gain(\mathcal{P}_V) = {\sum}_{h} \mathbf{w}_{h} \Phi_h(\mathcal{P}_V) \\
    end(\mathcal{P}) = j \\
    cost(\mathcal{P}) = cost(\mathcal{P}^{-}) + T_{i,j} + s_j. \\
  \end{array} \right.
\end{equation}

\textbf{Compact states $\mathbb{C}$.}
$\mathcal{P}_{V}$ and $Gain(\mathcal{P}_V)$ depend on the POI set of the route $\mathcal{P}$ but are independent of how the POIs are ordered. 
Hence, we group all open routes sharing the same $\mathcal{P}_{V}$ as a \emph{compact state}, denoted as $\mathbb{C}$, and let $\mathbb{C}_L$ denote the list of open routes having $\mathbb{C}$ as the POI set.
$\mathbb{C}$ is associated with the following fields:
\begin{equation}\label{eq:hash}
\left\{
	\begin{array}{l}
		Gain(\mathbb{C}): \mbox{the gain of routes grouped by $\mathbb{C}$}\\
		\mathbb{C}_{L}:    \forall \mathcal{P} \in \mathbb{C}_{L}, end(\mathcal{P}), cost(\mathcal{P}).  \\
	\end{array}
\right.
\end{equation}
These information is cached in a hash map with $\mathbb{C}$ as the key.

We assume that the POIs in $\mathcal{V}_{Q}$ are arranged in the lexicographical order of POI IDs. 
The compact states are enumerated as the subsets of $\mathcal{V}_{Q}$.
$x$ is included in every compact state, so we omit $x$.

Figure \ref{fig:tree} shows a compact state enumeration tree for $\mathcal{V}_{Q} = \{A, B, C, D\}$, excluding $x$ and $y$.
Each capital letter represents a POI, each node represents a compact state.
We define the set of POIs that precede $i$, in the above order, in a POI set as the \emph{prefix} of a POI $i$, e.g., prefix of $C$ is $\{AB\}$.
The compact states are generated in a specific \emph{prefix-first depth-first} manner so that longer open routes are extended from earlier computed shorter ones.
Initially, the root is the empty set $\emptyset$.
A child node $\mathbb{C}$ of the current node $\mathbb{C}^-$ is generated by appending a POI $i$ that precedes any POIs in $\mathbb{C}^-$ to the front of $\mathbb{C}^-$, and all child nodes are arranged by the order of $i$.
For example, Node 7 $\{ABC\}$ is generated as a child node of Node 6 $\{BC\}$ by appending $A$ to the front of $\{BC\}$ because $A$ precedes $B$ and $C$.

\begin{figure}[]%
        \centering
        \includegraphics[width=3.1in]{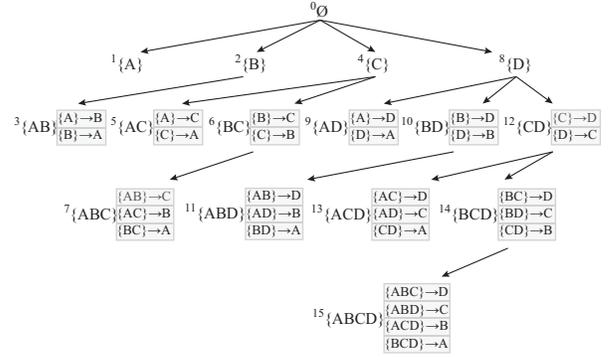}
        \caption{The compact state enumeration tree for \textbf{PACER}. 
        The number indicates the order of enumeration.
        }
        \label{fig:tree}
\end{figure}

At node $\mathbb{C}$, the open routes in $\mathbb{C}_L$ are generated by extending the cached routes in every compact state
$\mathbb{C}^{-j} = \mathbb{C} \setminus \{j\}$ where $j \in \mathbb{C}$. There are $|\mathbb{C}|$ such $\mathbb{C}^{-j}$.
We generate each route $\mathcal{P} = \mathcal{P}^{-} \rightarrow j$ by selecting the routes $\mathcal{P}^{-}$ from each $\mathbb{C}^{-j}_L$ and append $j$ at the end, and 
compute the gain and cost of $\mathcal{P}$ based on the accessed information for $\mathbb{C}^{-j}$ from the hash map. $\mathcal{P}$ is kept in $\mathbb{C}_L$ if it is feasible.

For example, to generate the open routes at the node $\{ABC\}$, we access the cached open routes at nodes $\{AB\}$, $\{AC\}$ and $\{BC\}$ and append the missing POI.
$\{AB\} \rightarrow C$ represents all the open routes ended with $C$ and the first two POIs in any order, i.e., $x\rightarrow A \rightarrow B \rightarrow C$ and $x\rightarrow B \rightarrow A\rightarrow C$. 
Note that it materializes only the current expanded branch of the tree, instead of the entire tree.

A closed route $\mathcal{P} \rightarrow y$ for each $\mathcal{P} \in \mathbb{C}_{L}$ is used to update the top-$k$ routes $topK$. If $\mathbb{C}_{L}$ is empty, this compact state is not kept. If no compact state is expandable, we stops the enumeration and yield the final $topK$.
Note that each $\mathbb{C}_{L}$ can include $|\mathbb{C}|!$ open routes and enumerating all the routes can be very expensive. We present two strategies to prune unpromising routes.

\subsection{Cost-based Pruning Strategy} \label{sec:feature_prune0}

Consider two feasible open routes $\mathcal{P}$ and $\mathcal{P}'$. We say $\mathcal{P}$ \emph{dominates} $\mathcal{P}'$ if $\mathcal{P}_V = \mathcal{P}'_V$, $end(\mathcal{P}) = end(\mathcal{P}')$ and $cost(\mathcal{P}) \leq cost(\mathcal{P}')$.
Because, if a route $\mathcal{P}'\rightarrow \hat{\mathcal{P}}$ is feasible, 
$\mathcal{P}\rightarrow \hat{\mathcal{P}}$ with the same extension $\hat{\mathcal{P}}$ must be also feasible and $cost(\mathcal{P}\rightarrow \hat{\mathcal{P}}) \leq cost(\mathcal{P}' \rightarrow \hat{\mathcal{P}})$.

\textbf{Pruning-1: cost dominance pruning.}
Leveraging the above dominance relationship, at the compact state $\mathbb{C}$, when generating $\mathcal{P} = \mathcal{P}^{-} \rightarrow j$ for a given $j$, we only select the open route $\mathcal{P}^{-}$ from $\mathbb{C}^{-j}_{L}$
such that $\mathcal{P}$ is feasible and has the least cost, thus, dominates all other routes  $\mathcal{P}'^{-} \rightarrow j$ with $\mathcal{P}'^{-}$ from $\mathbb{C}^{-j}_{L}$.
This reduces $|\mathbb{C}|!$ open routes to at most $|\mathbb{C}|$ dominating open routes at the compact state $\mathbb{C}$, one for each $j$ in $\mathbb{C}$, without affecting the optimality.
We call this strategy \emph{cost dominance pruning}.

For example, $\{AB\} \rightarrow C$ on node 7 $\{ABC\}$ now represents only one open route with the least cost chosen from $A \rightarrow B \rightarrow C$ and $B \rightarrow A \rightarrow C$. Note that Pruning-1 is a subtree pruning, e.g., if $A \rightarrow B$ on node 3 is pruned, all the open routes starting with $A \rightarrow B$, such as $A \rightarrow B \rightarrow C$ on node 7 and 
$A \rightarrow B \rightarrow D$ on node 11, will never be considered.

Though all dominated open routes are pruned, many of the remaining dominating open routes are still unpromising to lead to the top-$k$ closed routes. This further motivates our next strategy.

\subsection{Gain-based Pruning Strategy} \label{sec:feature_prune}

We can extend a dominating open route $\mathcal{P}$ step by step using the remaining budget $\Delta{b} = b - cost(\mathcal{P})$ into a closed route $\mathcal{P} \rightarrow \hat{\mathcal{P}}$. The POIs used for extension at each step should be reachable from the current $end(\mathcal{P})$, therefore, chosen from the set
\begin{equation}\label{eqn:U}
\mathcal{U} = \{i|T_{end(\mathcal{P}),i} + s_i + T_{i,y} + s_y \leq \Delta{b} \},
\end{equation}
where $i$ is an unvisited POI other than $y$. $T_{end(\mathcal{P}),i}$ and $T_{i,y}$ can be computed through $\mathsf{HI_Q}$. 
$\mathcal{P} \rightarrow \hat{\mathcal{P}}$ has gain $Gain(\mathcal{P}_V \cup \hat{\mathcal{P}}_V)$.
Then the \emph{marginal gain} by concatenating $\hat{\mathcal{P}}$ to the existing $\mathcal{P}$ is
\begin{equation} \label{eq:exp_gain}
\Delta{Gain(\hat{\mathcal{P}}_V | \mathcal{P}_V)} = Gain(\mathcal{P}_V \cup \hat{\mathcal{P}}_V) - Gain(\mathcal{P}_V).
\end{equation}
Let $\mathcal{P} \rightarrow \hat{\mathcal{P}}^*$ denote the $\mathcal{P} \rightarrow \hat{\mathcal{P}}$ with the highest gain. If $\mathcal{P} \rightarrow \hat{\mathcal{P}}^*$ ranks lower than the current $k$-th top routes $topK[k]$, $\mathcal{P}$ is not promising and all the open routes extended from $\mathcal{P}$ can be pruned.

\textbf{Pruning-2: marginal gain upper bound pruning}. However, finding $\hat{\mathcal{P}}^*$ is as hard as finding an optimal route from scratch, so we seek to estimate an \emph{upper bound} $UP$
of the marginal gain $\Delta{Gain(\hat{\mathcal{P}}_V | \mathcal{P}_V)}$, 
such that if $Gain(\mathcal{P}_V) + UP$ is less than
the gain of $topK[k]$,
$\mathcal{P}$ is not promising,
thus, $\mathcal{P}$ and all its extensions can be pruned without affecting the optimality.
We call this \emph{marginal gain upper bound pruning}.
As more routes are enumerated, the gain of $topK[k]$ increases and this pruning becomes more powerful.


The challenge of estimating $UP$ is to estimate the cost of the extended part $\hat{\mathcal{P}}$ without knowing the order of the POIs.
Because $\Delta{Gain(\hat{\mathcal{P}}_V | \mathcal{P}_V)}$ is independent of the POIs' order, 
we can ignore the order and approximate the ``route cost'' by a ``set cost'', i.e., 
the sum of some cost $c(i)$ of each POI $i \in \hat{\mathcal{P}}_V$, where $c(i)$ is 
no larger than $i$'s actual cost when it is included in $\hat{\mathcal{P}}$. We define 
$c(i)$ as:
\begin{equation}
c(i) = s_i + min(t_{j,i})/2 + min(t_{i,k})/2,
\end{equation}
where $t_{j,i}$ is the cost on an in-edge $e_{j,i}$ and $t_{i,k}$ is the cost on an out-edge $e_{i,k}$.
As the order of POIs is ignored, it is easy to verify that 
$min$ ensures the above property of $c(i)$.
The destination $y$ is ``one-sided'', i.e.,  $c(y)=s_y + min(t_{j,y})/2$.
To make a tighter cost approximation, we also count the half out-edge cost $min(t_{end(\mathcal{P}),k})/2$ for $end(\mathcal{P})$. 


Then, $UP$ is exact the solution, i.e., the maximum $\Delta{Gain(S^* | \mathcal{P}_V)}$, to the following optimization problem:
\begin{equation}\label{eq:set_diff_opt}
\begin{split}
\underset{ S \subseteq \mathcal{U} \cup \{y\}}{\max} \Delta{Gain(S | \mathcal{P}_V)} \
 s.t. \ {\sum}_{i \in S} c(i)  \leq B,
\end{split}
\end{equation}
where $\mathcal{U}$ is defined in Eqn. (\ref{eqn:U}) and $B = \Delta{b} - min(t_{end(\mathcal{P}),k})/2$. Note that $S$ should include $y$ because $end(\hat{\mathcal{P}}) = y$.
As $c(i)$ and $c(end(\mathcal{P}))$ are no larger than their actual costs,  $\Delta{Gain(S^* | \mathcal{P}_V)}\geq \Delta{Gain(\hat{\mathcal{P}}_V | \mathcal{P}_V)}$ for any $\hat{\mathcal{P}}$.
Thus, using $\Delta{Gain(S^* | \mathcal{P}_V)}$ as $UP$ never loses the optimality. 
To solve Eqn. (\ref{eq:set_diff_opt}), we first show the properties of the marginal gain function $\Delta{Gain}$.

%

\begin{thm} \label{thm:sub_increment}
The marginal gain function $\Delta{Gain}$ as defined in Eqn. (\ref{eq:exp_gain}) is nonnegative, monotone and submodular.
\end{thm}
\begin{proof}
We only show that
$\Delta{Gain}$ is submodular.
According to \cite{krause2012submodular}, if a set function $g:2^V \rightarrow \mathbb{R}$ is submodular, and $X,Y \subset V$ are disjoint, the \emph{residual} function $f: 2^Y \rightarrow \mathbb{R}$ defined as
$f(S) = g(X \cup S) - g(X)$ is also submodular.
Since $Gain$ is submodular (Theorem \ref{thm:sub}) and since $\mathcal{P}_V, \mathcal{U} \subset \mathcal{V}$ are disjoint, $\Delta{Gain(\hat{\mathcal{P}}_V | \mathcal{P}_V)} = Gain(\mathcal{P}_V \cup \hat{\mathcal{P}}_V) - Gain(\mathcal{P}_V)$ is residual on $\hat{\mathcal{P}}_V$, thus, is submodular.
\end{proof}

Apparently, Eqn. (\ref{eq:set_diff_opt}) is a submodular maximization problem subject to a knapsack constraint, which unfortunately is also NP-hard \cite{krause2012submodular}. 
Computing $\Delta{Gain(S^* | \mathcal{P}_V)}$ is costly, thus, we consider estimating its upper bound.

One approach, according to \cite{sviridenko2004note}, is to run a $\Omega(B|\mathcal{U}|^4)$ time ($B$ is defined in Eqn. (\ref{eq:set_diff_opt})) greedy algorithm in \cite{khuller1999budgeted} to obtain an approximate solution $\Delta{Gain(S'| \mathcal{P}_V)}$ for the above problem with approximation ratio of $1 - e^{-1}$, then the upper bound of $\Delta{Gain(S^* | \mathcal{P}_V)}$ is achieved by $ \Delta{Gain(S'| \mathcal{P}_V)} / (1 - e^{-1})$.
A less costly version of this algorithm runs in $\mathcal{O}(B|\mathcal{U}|)$ but its approximation ratio is $\frac{1}{2}(1 - e^{-1})$.


Compared with the above mentioned \emph{offline} bounds, i.e., $1 - e^{-1}$ and $\frac{1}{2}(1 - e^{-1})$ that are stated in advance before running the actual algorithm, the next theorem states that we can instead use the submodularity to acquire a much tighter \emph{online} bound.

\begin{thm} \label{th:online_bound}
For each POI $i \in \mathcal{U} \cup \{y\}$, let $\delta_i = \Delta{Gain(\{i\} | \mathcal{P}_V)}$. Let $r_i = \delta_i/c(i)$, and let $i_1, \cdots, i_m$ be the sequence of these POIs with $r_i$ in decreasing order.
Let $l$ be such that $C = {\sum}_{j=1}^{l-1}c(i_j)$ $\leq$ $B$ and ${\sum}_{j=1}^{l}c(i_j) > B$. Let $\lambda = (B-C)/c(i_l)$. Then
\begin{equation}\label{eq:upper_bound}
UP = {\sum}_{j=1}^{l-1}\delta_{i_j} + \lambda \delta_{i_l} \geq \Delta{Gain(S^* | \mathcal{P}_V)}.
\end{equation}
\end{thm}
\begin{proof}
\cite{leskovec2007cost} showed a theorem that a tight online bound for arbitrary given solution $\hat{\mathcal{A}}$ (obtained using any algorithm) to a constrained submodular maximization problem can be got to measure how far $\hat{\mathcal{A}}$ is from the optimal solution.
By applying \cite{leskovec2007cost} to the problem in Eqn. (\ref{eq:set_diff_opt}) and let $\hat{\mathcal{A}} = \emptyset$, Theorem \ref{th:online_bound} is deduced.
\end{proof}

By this means, $UP$ is computed without running a greedy algorithm. We also empirically proved that this online bound in Eqn. (\ref{eq:upper_bound}) outperforms the offline bounds on both tightness and computational cost.
Thus, we finally choose the online bound.

\subsection{Algorithm}\label{sec:alg}




\begin{algorithm}[t]
\SetKwInOut{Input}{Parameters}
\SetKwInOut{Glb}{Globals}
\SetKwInOut{Output}{Output}
\SetKwFunction{Check}{Check}
\SetKwFunction{PACER}{PACER}
\SetKwFunction{Sat}{Sat}
\SetKwFunction{Reach}{Reach}
\SetKwFunction{UpdState}{UpdateState}
\SetKwFunction{Find}{NonDominatedRoute}
\SetKwFunction{UpdK}{UpdateTopK}
\SetKwFunction{Pre}{Pref}
\SetKwFunction{Upper}{UpperBoundMarginal}

	\Glb {$Q = (x,y, b,\mathbf{w},\boldsymbol{\theta}, \Phi)$, $\mathcal{V}_Q$, $\mathsf{FI_Q}$ and $\mathsf{HI_Q}$ to compute $Gain(\mathbb{C})$ and $cost(\mathcal{P})$, and $k$}
	\Input {compact state $\mathbb{C}^{-}$ and the set of POIs $I$ 
	}
	\Output {a priority queue $topK$}

	\ForAll{POI $i$ in set $I$ in order}
	{ \label{li:for_v}
		$\mathbb{C} \leftarrow$  $\{i\}$ $\cup$ $\mathbb{C}^{-}$;\\ \label{li:pre_v}
		compute $Gain(\mathbb{C})$; \\ \label{li:com_gain}
		\ForAll{POI $j$ in $\mathbb{C}$}
	 	{ \label{li:sch_2}
	 		$\mathbb{C}^{-j} \leftarrow$ $\mathbb{C} \setminus \{j\}$; \\ \label{li:key_2}
	 		$\mathcal{P}^{-} \leftarrow$ the dominating route in $\mathbb{C}^{-j}_{L}$ such that $cost(\mathcal{P}^{-} \rightarrow j)$ is minimum; \ \tcp*[h]{prune-1} \\ \label{li:find}
			$\mathcal{P} \leftarrow \mathcal{P}^{-} \rightarrow j$; \\
            \If{$cost(\mathcal{P} \rightarrow y) \leq b$}
            {\label{li:sat}
				Compute $UP$ using Eqn. (\ref{eq:upper_bound}); \\ \label{li:upper}
	 			\If{$Gain(\mathbb{C}) + UP \geq Gain(topK[k])$}
	 			{
	 				insert route $\mathcal{P}$ into $\mathbb{C}_{L}$; \ \tcp*[h]{prune-2} \\ \label{li:insert}
	 			}
            }		 		
	 	}
	 	$\UpdK(\mathbb{C}_{L}$, $topK)$;  \\ \label{li:opt_p}
	 	$\PACER(\mathbb{C}, $ prefix of $i$ in $I)$; \\ \label{li:rec_2}
	}

\caption{\texttt{PACER}($\mathbb{C^{-}}, I$) (Recursive funcion)}\label{algo:2}
\end{algorithm}

Algorithm \ref{algo:2} incorporates the above enumeration and pruning strategies.
Given the global variables, \texttt{PACER($\mathbb{C}^{-},I$)} recursively enumerates the subtree at the current compact state $\mathbb{C}^{-}$ with the POI set $I$ available for extending $\mathbb{C}^{-}$, and finally return the $k$ best routes in $topK$.
The initial call is \texttt{PACER($\emptyset,\mathcal{V}_{Q}$)}, when only $x$ is included.

As explained in Section \ref{sec:enumeration}, Line \ref{li:for_v} - \ref{li:com_gain} extends $\mathbb{C}^{-}$ by each $i$ in the set $I$ in order, creating the child node $\mathbb{C}$ and computing $Gain(\mathbb{C})$.
Lines \ref{li:sch_2} - \ref{li:insert} generate the dominating and promising open routes $\mathbb{C}_{L}$.
Specifically, for each $j \in \mathbb{C}$ selected as the ending POI,
Line \ref{li:key_2} - \ref{li:find} find the dominating route $\mathcal{P}^{-}$ from the previously computed $\mathbb{C}^{-j}_{L}$.
This corresponds to Pruning-1.
Only when the new open route $\mathcal{P}$ is feasible,
Pruning-2 is applied to check 
if $\mathcal{P}$ has a promising gain, and if so, $\mathcal{P}$ is inserted into $\mathbb{C}_{L}$ (Lines \ref{li:sat} - \ref{li:insert}).
After $\mathbb{C}_{L}$ is finalized, it selects an open route $\mathcal{P}$ in $\mathbb{C}_{L}$ such that $\mathcal{P} \rightarrow y$ has the least cost to update $topK$ (Line \ref{li:opt_p}).
The information of the new compact state $\mathbb{C}$, as in Eqn. (\ref{eq:hash}), is added to the hash map.
At last, $\mathbb{C}$ is extended recursively with the POIs
in the prefix of $i$ in current $I$ (Line \ref{li:rec_2}).


\textbf{Summary of the properties of PACER.} (1) PACER works for \textbf{any} nonnegative, monotone, and submodular $Gain$ function so as to deal with the personalized diversity requirement.
(2) Open routes are enumerated as compact states in a prefix-first depth-first order to construct open routes incrementally, i.e., \textbf{dynamic programming}. (3) With \textbf{Pruning-1}, we compute at most $|\mathbb{C}|$ dominating feasible open routes at each compact state $\mathbb{C}$, instead of $|\mathbb{C}|!$ routes.
(4) \textbf{Pruning-2} further weeds out the dominating feasible open routes not having a promising estimated maximum gain. 

\subsection{Complexity Analysis} \label{sec:complex}

We measure the computational complexity by the number of routes examined.
Two main factors affecting this measure are the size of the POI candidate set, i.e., $|\mathcal{V}_Q|$ denoted by $n$, and the maximum length of routes examined (excluding $x$ and $y$), i.e., the maximum $|\mathcal{P}|$ denoted by $p$.
$p \ll n$.
We analyze PACER relatively to
the brute-force search and a state-of-the-art approximation solution.

\textbf{PACER.} The compact states on the $l$-th level
of the enumeration tree (Figure \ref{fig:tree}) compute the routes containing $l$ POIs; thus, there are at most ${n \choose l}$ compact states on level $l$. And thanks to Pruning-1, each compact state represents at most $l$ dominating open routes. 
There are $n$ dominating open routes with single POI on level $l = 1$.
Starting from $l = 2$, to generate each dominating open route on level $l$, we need to examine $(l-1)$ sub-routes having the same set of POIs as the prefix and add the same ending to find the dominating one, according to the cost dominance pruning. 
Hence, with $p \ll n$ and the Pascal's rule \cite{burton2006elementary}, the number of routes examined is at most
\begin{equation} \label{eq:comp_pacer}
\begin{split}
& n + \sum_{l=2}^{p} l(l-1){n \choose l} = n + n(n-1) \sum_{l=2}^{p}{n-2 \choose l-2} \approx n(n-1)({n-2 \choose p-2}   \\
& +  {n-2 \choose p-3}) = n(n-1) {n-1 \choose p-2} = \frac{n-1}{(n-p+1)(p-2)!} \frac{n!}{(n-p)!}.
\end{split}
\end{equation}
Therefore, the computation cost of PACER is $\mathcal{O}(\frac{1}{(p-2)!} \frac{n!}{(n-p)!})$ with $p \ll n$. If Pruning-2 is also enabled and it prunes the $\gamma$ percent of the routes examined by PACER with Pruning-1, the computation cost of PACER is $\mathcal{O}((1 - \gamma) \frac{1}{(p-2)!} \frac{n!}{(n-p)!})$.

\textbf{Brute-force algorithm (BF).} The brute-force algorithm based on the breadth-first expansion examines $\mathcal{O}(\frac{n!}{(n-p)!})$ routes.

\textbf{Approximation algorithm (AP).}
\cite{chekuri2005recursive} proposed a \emph{quasi-polynomial time} approximation algorithm for the Orienteering Problem. We modified AP
to solve our problem.
It uses a recursive binary search to 
produce a single route with the approximation ratio $\lceil \log p\rceil + 1$ and runs in 
$\mathcal{O}((n \cdot OPT \cdot \log b)^{\log p})$, where $OPT$ and $b$ are the numbers of discrete value for an estimated optimal Gain and for the budget, respectively.
The cost is expensive if $b$ or $OPT$ has many discrete values.
For example, for $b = 512$ minutes, $n = 50$, $p = 8$
and $OPT = 10.0$ (100 discrete values with the single decimal point precision),
the computation cost 
is $(50 \cdot 100 \cdot \log 512)^{\log 8} = 9.11 \times 10^{13}$.
Lower precision leads to smaller computation cost, but also lower accuracy.
\cite{singh2007efficient} noted that AP took more than $10^4$ seconds for a small graph with 22 nodes.
Compared with AP, the computation cost of PACER with Pruning-1 given by Eqn. (\ref{eq:comp_pacer}) is only $50 \times 49 \times {49 \choose 6} = 3.43 \times 10^{10}$.
This cost is further reduced by Pruning-2.
PACER finds the optimal top-$k$ routes whereas AP only finds single approximate solution.
We will experimentally compare PACER with AP.

\section{Heuristic Solutions} \label{sec:heu}
PACER remains expensive for a large  budget $b$ and a large POI candidate set $\mathcal{V}_Q$.
The above approximation algorithm is not scalable.
Hence, we design two heuristics when such extreme cases arise.

\textbf{State collapse heuristic.}
The cost dominance pruning in PACER keeps at most $l$ open routes for a compact state representing a set of $l$ POIs (excluding $x$ and $y$). A more aggressive pruning is to keep only one open route having the least cost at each compact state, with the heuristic that this route likely visits more POIs. 
We denote this heuristic algorithm by \textbf{PACER-SC}, where SC stands for ``State Collapsing''.
Clearly, PACER-SC trades optimality for efficiency, but it inherits many nice properties from PACER and Section \ref{sec:perfomance} will show that it usually produces $k$ routes with quite good quality.

Analogous to the complexity analysis for PACER in Section \ref{sec:complex}, with $p \ll n$, PACER-SC examines no more than 
$\sum_{l=1}^{p} l{n \choose l} \approx n{n \choose p-1}$ routes, which is around $1/p$ of that for PACER. 

\textbf{Greedy algorithm.}
PACER-SC's computation complexity  remains exponential in the route length $p$.
Our next greedy algorithm runs in polynomial time. It starts with the initial route $x \rightarrow y$ and iteratively inserts an unvisited POI $i$ to the current route to maximize the  marginal gain/cost ratio
\begin{equation} \label{eq:gr_ratio}
\frac{Gain(\{i\} \cup \mathbb{C}) - Gain(\mathbb{C})}{s_i + T_{x,i} + T_{i,y}},
\end{equation}
where $\mathbb{C}$ denotes the set of POIs on the current route. It inserts $i$ between two adjacent POIs in the current route so that the total cost of the resulting route is minimized. The term $T_{x,i} + T_{i,y}$ constrains the selected POIs $i$ to be those not too far away from the two end points. The expansion process is repeated until the budget $b$ is used up. The algorithm only produces a single route and examines
$\mathcal{O}(pn)$ routes because each insertion will consider at most $n$ unvisited POIs.

\section{Experimental Evaluation} \label{se:exp}

All algorithms were implemented in C++ and were run on Ubuntu 16.04.1 LTS with Intel i7-3770 CPU @ 3.40 GHz and 16G of RAM. 

\subsection{Experimental Setup}
\subsubsection{Datasets}
We use two real-world datasets from \cite{zeng2015optimal}.
\textsl{Singapore} denotes the Foursquare check-in data collected in Singapore, and \textsl{Austin} denotes the Gowalla check-in data collected in Austin.
\textsl{Singapore} has 189,306 check-ins at 5,412 locations by 2,321 users, and \textsl{Austin} has 201,525 check-ins at 6,176 locations by 4,630 users. Same as suggested in \cite{cao2012keyword, zeng2015optimal}, we built an edge between two locations if they were visited on the same date by the same user. 
The locations not connected by edges were ignored.
We filled in the edge costs $t_{i,j}$ by querying the traveling time in minute using Google Maps API
under driving mode. 
The staying time $s_i$ were generated
following the Gaussian distribution, $s_i \sim \mathcal{N}(\mu,\sigma^{2})$, with $\mu = 90$ minutes and $\sigma = 15$.
The features are extracted based on the user mentioned keywords at check-ins, same to \cite{zeng2015optimal}.
We obtain the rating of a feature $h$ on POI $i$ by
\begin{equation} \label{eq:agg_rating}
\mathbf{F}_{i,h} = \min\{\frac{NC_{h}(i)}{1/|S_h| \times \sum_{j \in S_h} NC_{h}(j)} \times \frac{\beta}{2}, \beta\},
\end{equation}
where $NC_{h}(i)$ is the number of check-ins at POI $i$ containing the feature $h$, $S_h$ is the set of POIs containing $h$, $\beta$ is the maximum feature rating and is set to $\beta = 5$ for both data sets.
%
The calculation scales the middle value $\frac{\beta}{2}$ by the ratio of a POI's 
check-in count to the average check-in count on $h$.
Table \ref{tab:data} shows the descriptive statistics of the datasets after the above preprocessing.

\begin{table}[h]
\centering
\caption{Dataset statistics}
\label{tab:data}
\small
\begin{tabular}{cccccc}
\toprule
         & \textbf{\# POI} & \textbf{\# Edges} & \textbf{Average $t_{i,j}$}  & \textbf{\# Features}
         \\ \toprule
\textsl{Singapore} 	& 1,625		& 24,969	& 16.24 minutes		& 202\\ 
\textsl{Austin} 	& 2,609		& 34,340	& 11.12 minutes		& 252\\ \bottomrule
\end{tabular}
\vspace{-3pt}
\end{table}

Both datasets were used in \cite{zeng2015optimal}, which also studied a route planning problem. The datasets are not small considering the scenario for a daily trip in a city
where the user 
has a limited cost budget.
Even with 150 POIs to choose from, the number of possible routes consisting of 5 POIs can reach 70 billions.
Compared to our work, \cite{lim2017personalized} evaluated its itinerary recommendation methods using theme park data, where each park contains only 20 to 30 attractions.

\subsubsection{Algorithms} \label{sec:algos}
We compared the following algorithms. \textbf{BF} is the brute-force method (Section \ref{sec:complex}). \textbf{PACER+1} is our proposed optimal algorithm with only Pruning-1
enabled. \textbf{PACER+2} enables both Pruning-1 and Pruning-2. 
\textbf{PACER-SC} is the state collapse algorithm and \textbf{GR} is  the greedy algorithm in Section \ref{sec:heu}.
\textbf{AP} is the approximation algorithm proposed by \cite{chekuri2005recursive} (see Section \ref{sec:complex}).
\textbf{A*} is the A* algorithm proposed by \cite{zeng2015optimal}. Since
A* works only for its specific keyword coverage function, 
it is not compared until Section \ref{sec:Astar} where we adapt  their coverage function in our method.
To be fair, all algorithms use the indices in Section \ref{sec:overview} to speed up. Note that BF, PACER+1, PACER+2 and A* 
are exact algorithms, while PACER-SC, GR, and AP 
are greedy or approximation algorithms.

\subsubsection{Queries}
A query $Q$ has the six parameters $x,y, b,\mathbf{w},$ $\boldsymbol{\theta}, \Phi$.
For concreteness, we choose $\Phi_h$ in
Eqn. \eqref{eq:gain_power_2} with $\boldsymbol{\alpha}$ controlling the diversity of POIs on a desired route.
We assume $\boldsymbol{\theta}_h$ and $\boldsymbol{\alpha}_h$ are the same for all features $h$.
For \textsl{Singapore}, we set $x$ as Singapore Zoo and $y$ as Nanyang Technological University; 
and for \textsl{Austin}, we set $x$ as UT Austin and $y$ as Four Seasons Hotel Austin.

For each dataset, we generated 50 weight vectors $\mathbf{w}$ to model the feature preferences of 50 users as follows. For each $\mathbf{w}$, we draw $m$ features, where $m$ is a random integer in $[1,4]$, and the probability of selecting each feature $h$ is 
$\mathbf{Pr}(h) = \frac{\sum_{i \in S_h} NC_{h}(i)}{\sum_{h \in \mathcal{H}}\sum_{i \in S_h} NC_{h}(i)}$.
$NC_{h}(i)$ and $S_h$ are defined in Eqn. (\ref{eq:agg_rating}). 
Let $\mathcal{H}_Q$ be the set of selected features. For each $h \in \mathcal{H}_Q$, $\mathbf{w}_h = \frac{\sum_{i \in S_h} NC_{h}(i)}{\sum_{h \in \mathcal{H}_Q}\sum_{i \in S_h} NC_{h}(i)}$.

Finally, we consider $b \in \{4,5,\textbf{6},7,8,9\}$ in hours, $\boldsymbol{\theta} \in \{0, 1.25, \textbf{2.5},$ $3.75\}$, and $\boldsymbol{\alpha} \in \{0, \textbf{0.5}, 1, 2\}$ with the default settings in bold face. For each setting of $b, \boldsymbol{\theta}, \boldsymbol{\alpha}$, we generated 50 queries $Q=(x,y, b,\mathbf{w}, \boldsymbol{\theta}, \boldsymbol{\alpha})$ using the 50 vectors $\mathbf{w}$ above.
All costs are in minutes, therefore, $b=5$ specifies the budget of 300 minutes.



We first 
evaluate the performance of our proposed algorithms (Section \ref{sec:perfomance}), then we compare with the A* algorithm (Section \ref{sec:Astar}).

\subsection{Performance Study} \label{sec:perfomance}
\textbf{Evaluation metrics.}
As we solve an optimization problem, we evaluate \emph{Gain} for effectiveness, \emph{CPU runtime} and \emph{search space} (in the number of examined open routes) for efficiency.

For every algorithm, we evaluate the three metrics for processing a query, and
report the average for the 50 queries (i.e., vectors $\mathbf{w}$) under each setting of $(b,\boldsymbol{\theta},\boldsymbol{\alpha})$ chosen from the above ranges. 
GR and AP only find single route, thus, we first set $k=1$ to compare all algorithms, 
and discuss the impact of larger $k$ in Section \ref{sec:impact_k}.


Figures \ref{fig:sg} and \ref{fig:as} report the experiments for \textsl{Singapore} and \textsl{Austin}, respectively.
Each row corresponds to various settings of one of $b, \boldsymbol{\theta}, \boldsymbol{\alpha}$ while fixing the other two at the default settings.
OPTIMAL denotes the same optimal gain of PACER+2, PACER+1 and BF. 
We terminated an algorithm for a given query after
it runs for 1 hour or runs out of memory, and used the label 
beside a data point to indicate the percentage of finished queries.
If more than a half of the queries were terminated, no data point is shown.



\begin{figure}[t]  
        \hspace{0.2in}
        \begin{subfigure}[b]{0.25\textwidth}
                \includegraphics[width=\textwidth]{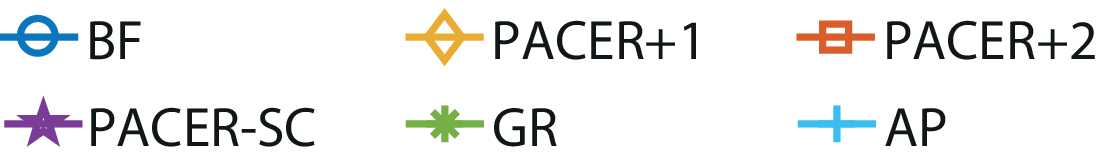}
        \end{subfigure}%
        ~ 
        \hspace{0.25in}
        \begin{subfigure}[b]{0.15\textwidth}
                \includegraphics[width=\textwidth]{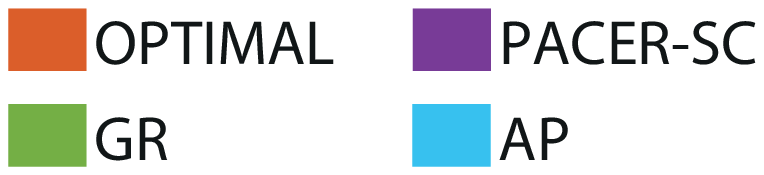}
        \end{subfigure}%

        \begin{subfigure}[b]{0.165\textwidth}
                \includegraphics[width=\textwidth]{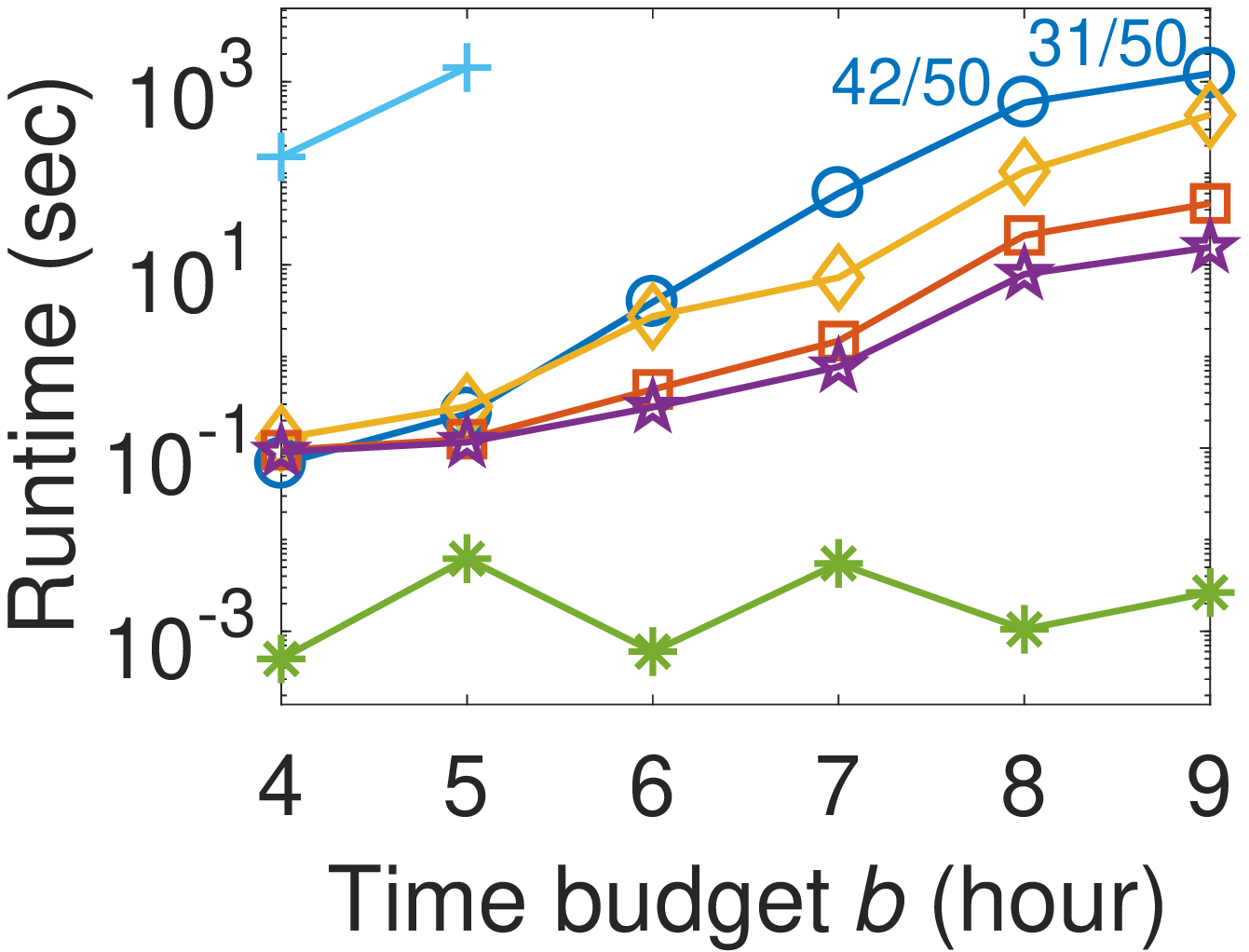}
                \caption{Runtime vs. $b$}
                \label{fig:sg_b_t}
        \end{subfigure}%
        ~ 
        \hspace{-0.15in}
        \begin{subfigure}[b]{0.165\textwidth}
                \includegraphics[width=\textwidth]{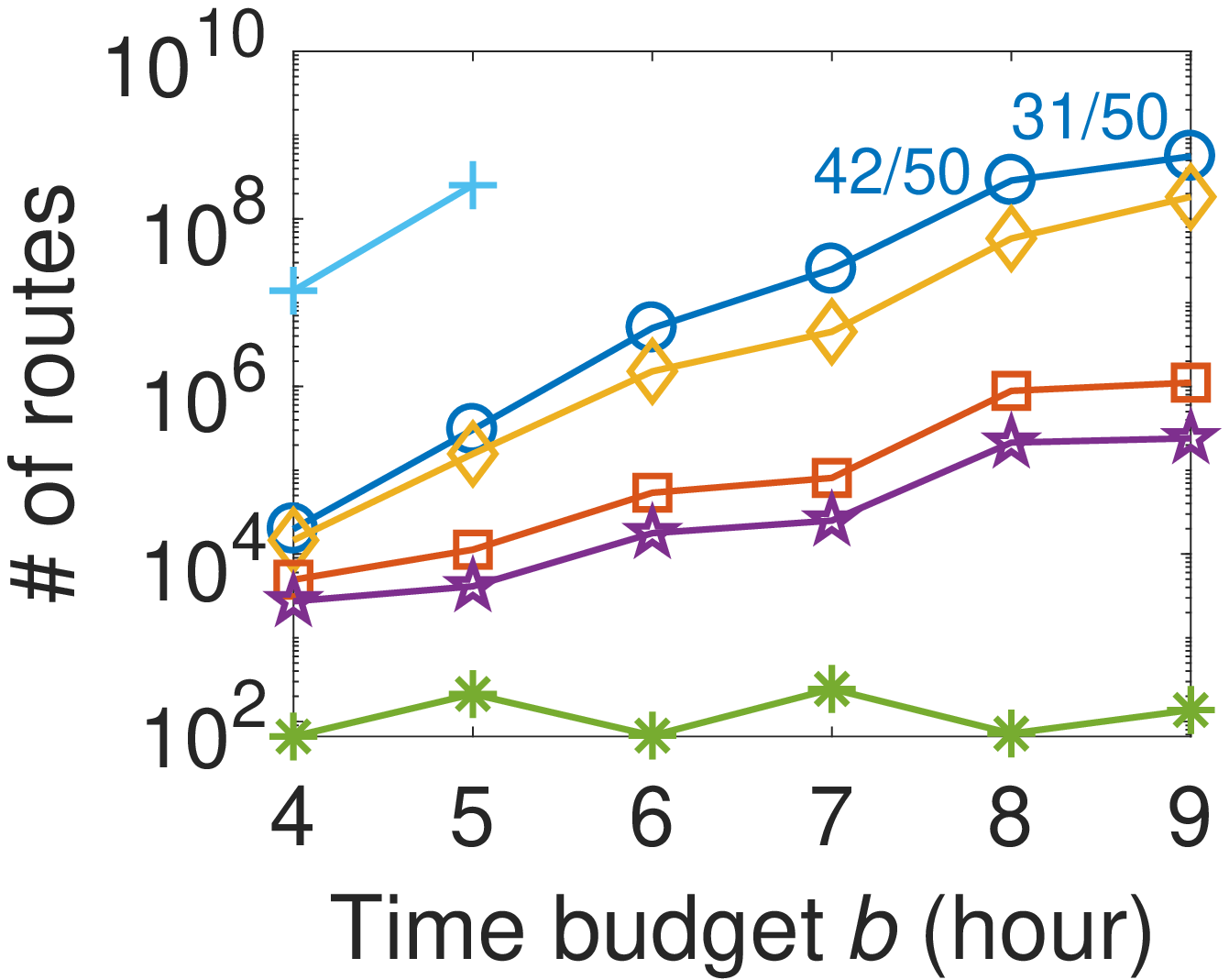}
                \caption{\# of routes vs. $b$}
                \label{fig:sg_b_s}
        \end{subfigure}
        ~ 
        \hspace{-0.15in}
        \begin{subfigure}[b]{0.165\textwidth}
                \includegraphics[width=\textwidth]{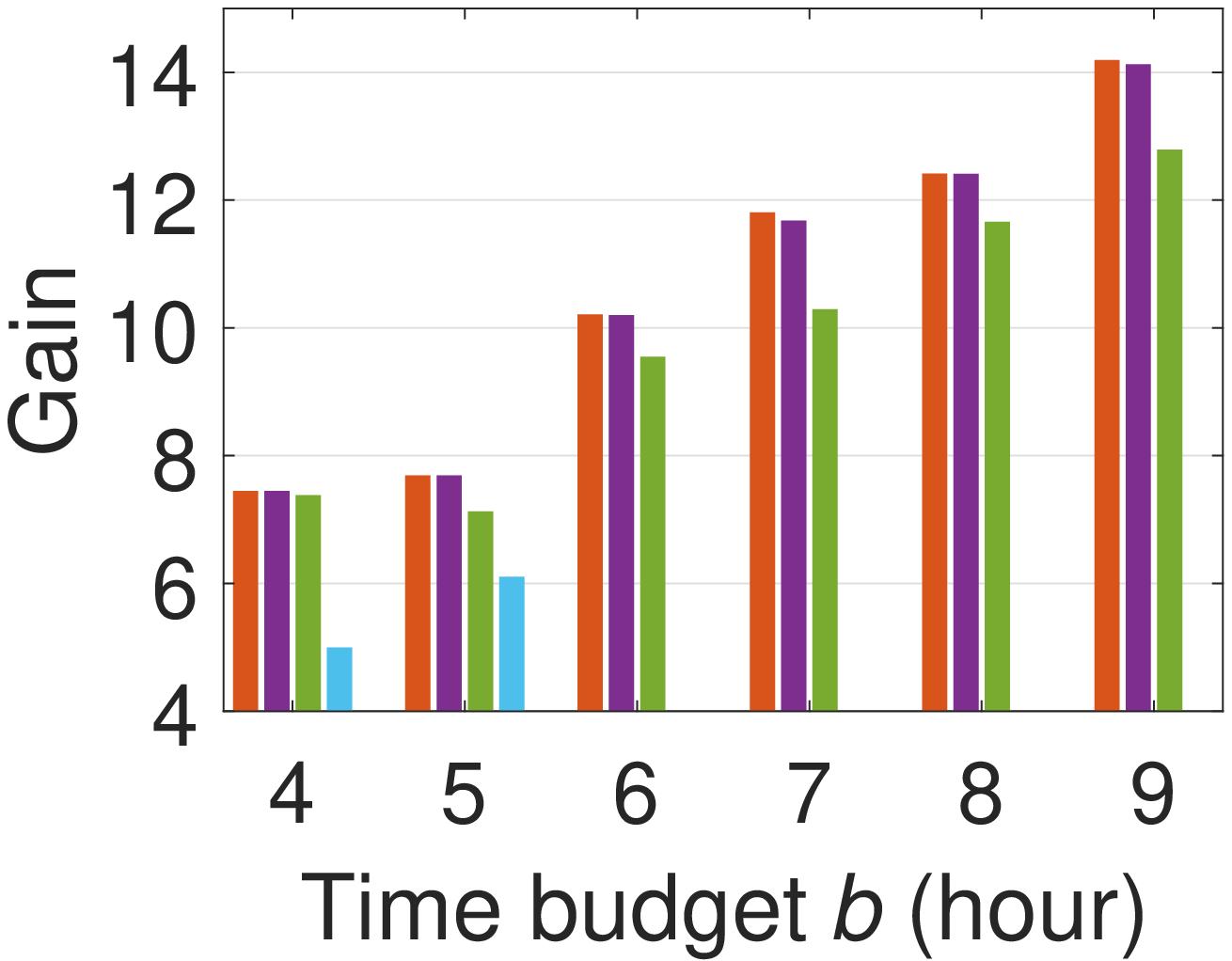}
                \caption{Gain vs. $b$}
                \label{fig:sg_b_g}
        \end{subfigure}

        \begin{subfigure}[b]{0.165\textwidth}
                \includegraphics[width=\textwidth]{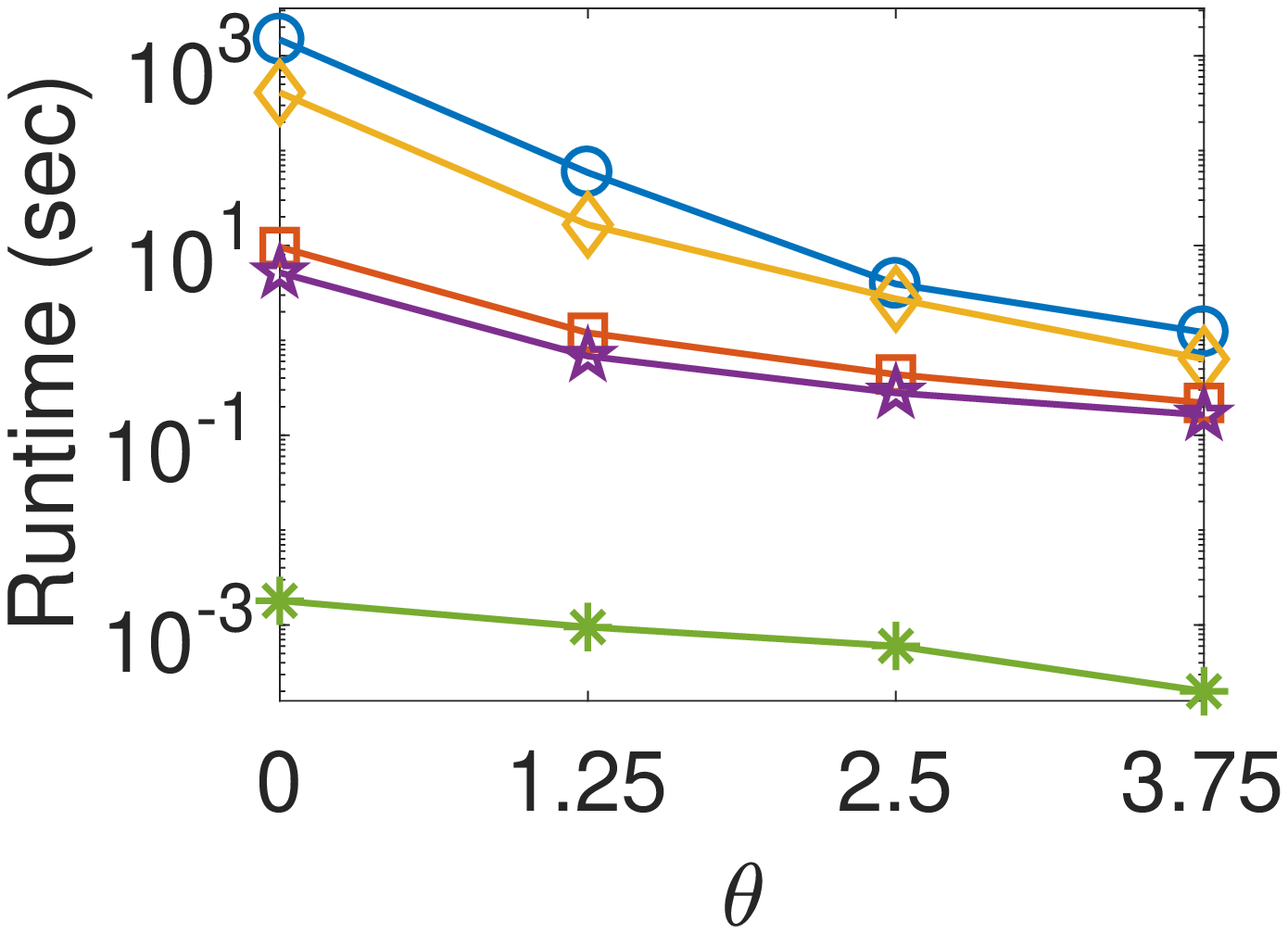}
                \caption{Runtime vs. $\boldsymbol{\theta}$}
                \label{fig:sg_theta_t}
        \end{subfigure}%
        ~ 
        \hspace{-0.15in}
        \begin{subfigure}[b]{0.165\textwidth}
                \includegraphics[width=\textwidth]{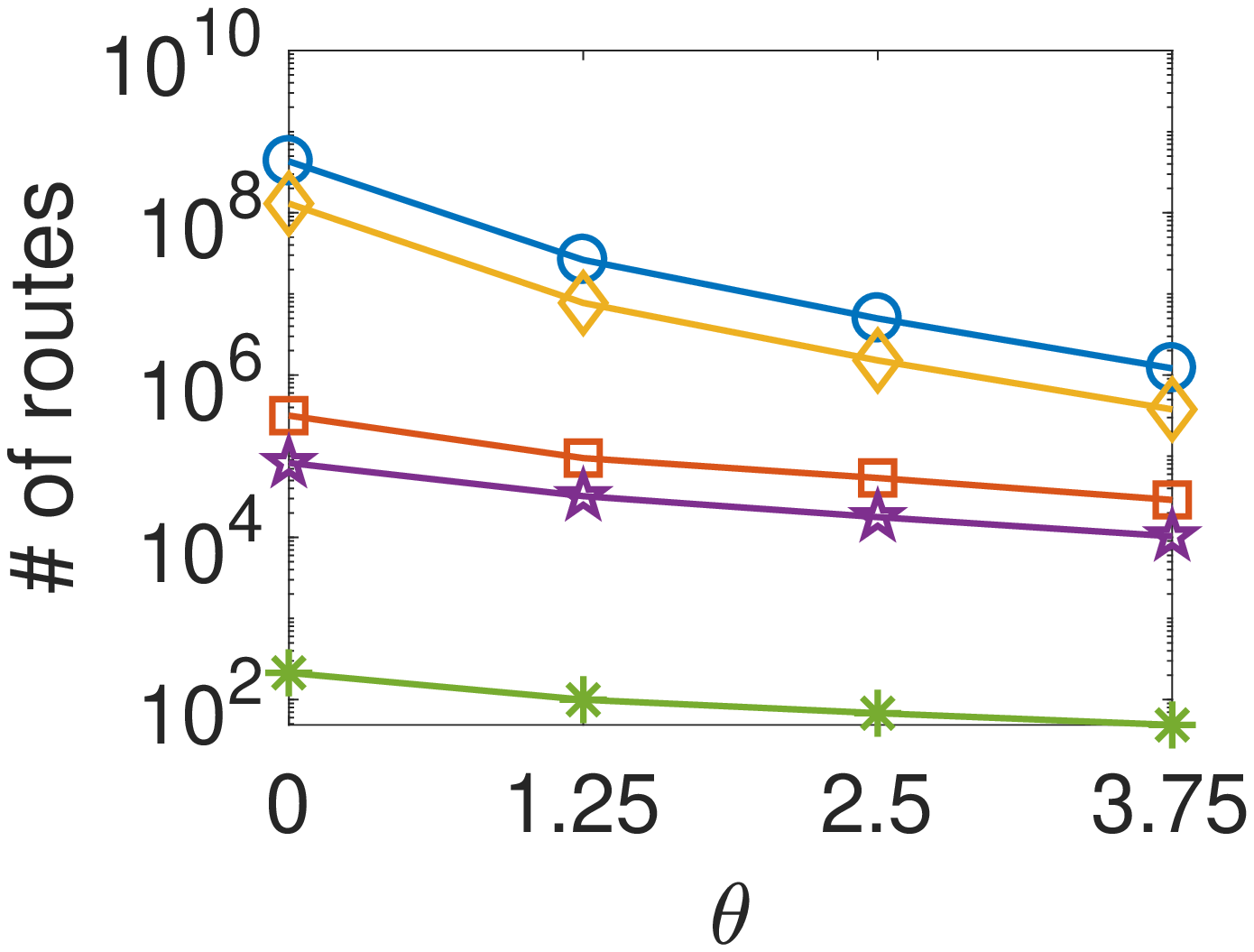}
                \caption{\# of routes vs. $\boldsymbol{\theta}$}
                \label{fig:sg_theta_s}
        \end{subfigure}
        ~ 
        \hspace{-0.15in}
        \begin{subfigure}[b]{0.165\textwidth}
                \includegraphics[width=\textwidth]{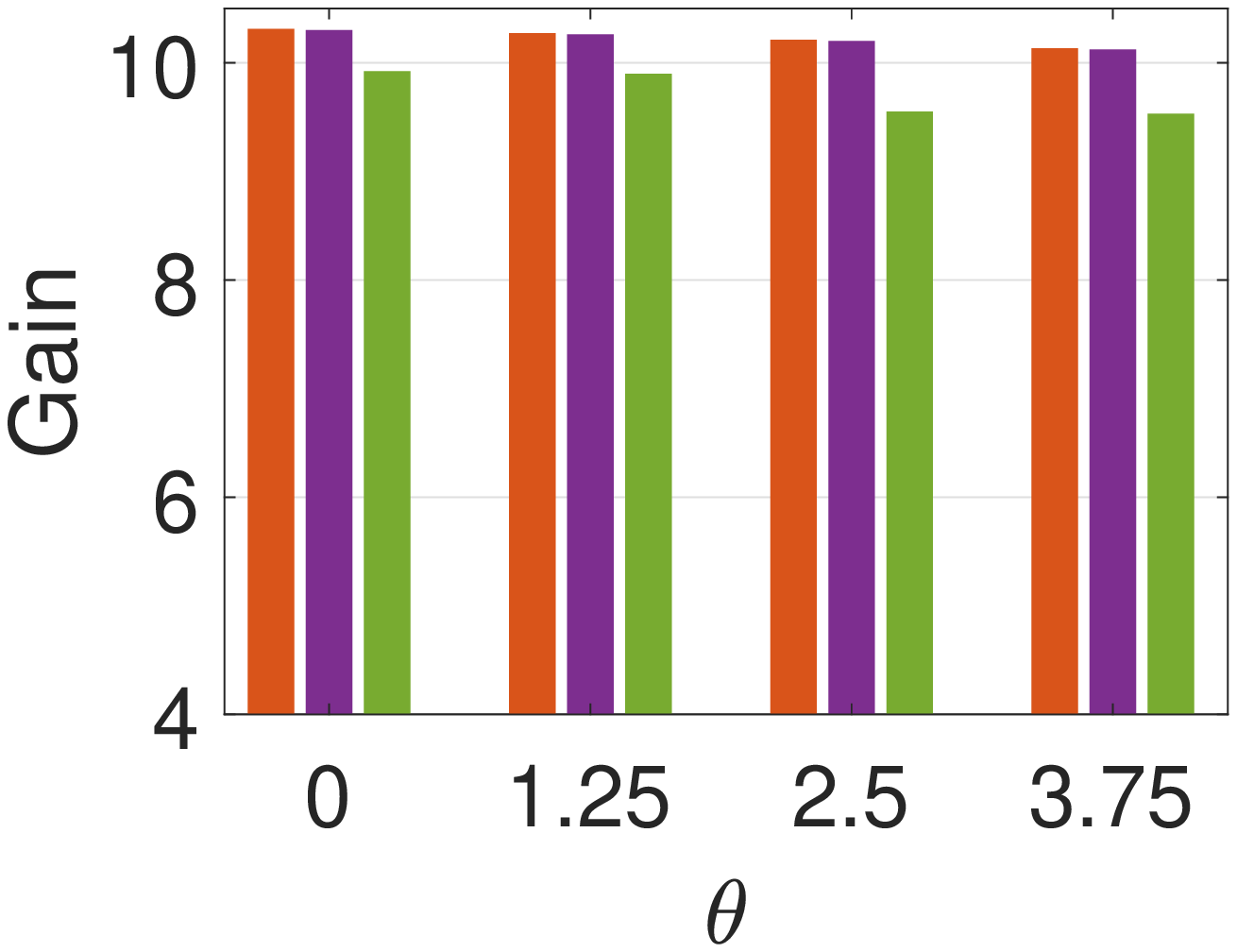}
                \caption{Gain vs. $\boldsymbol{\theta}$}
                \label{fig:sg_theta_g}
        \end{subfigure}

        \begin{subfigure}[b]{0.165\textwidth}
                \includegraphics[width=\textwidth]{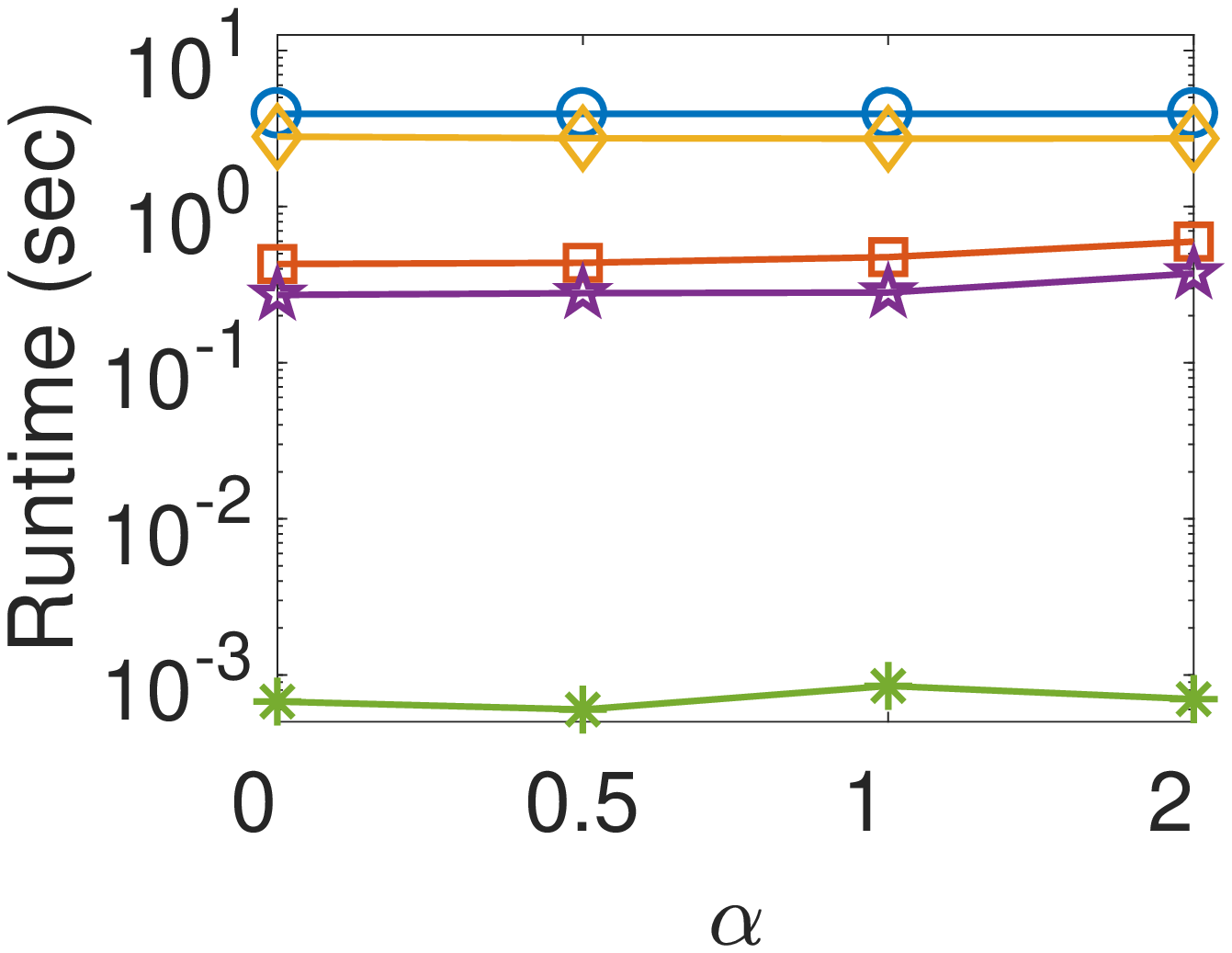}
                \caption{Runtime vs. $\boldsymbol{\alpha}$}
                \label{fig:sg_alpha_t}
        \end{subfigure}%
        ~ 
        \hspace{-0.15in}
        \begin{subfigure}[b]{0.168\textwidth}
                \includegraphics[width=\textwidth]{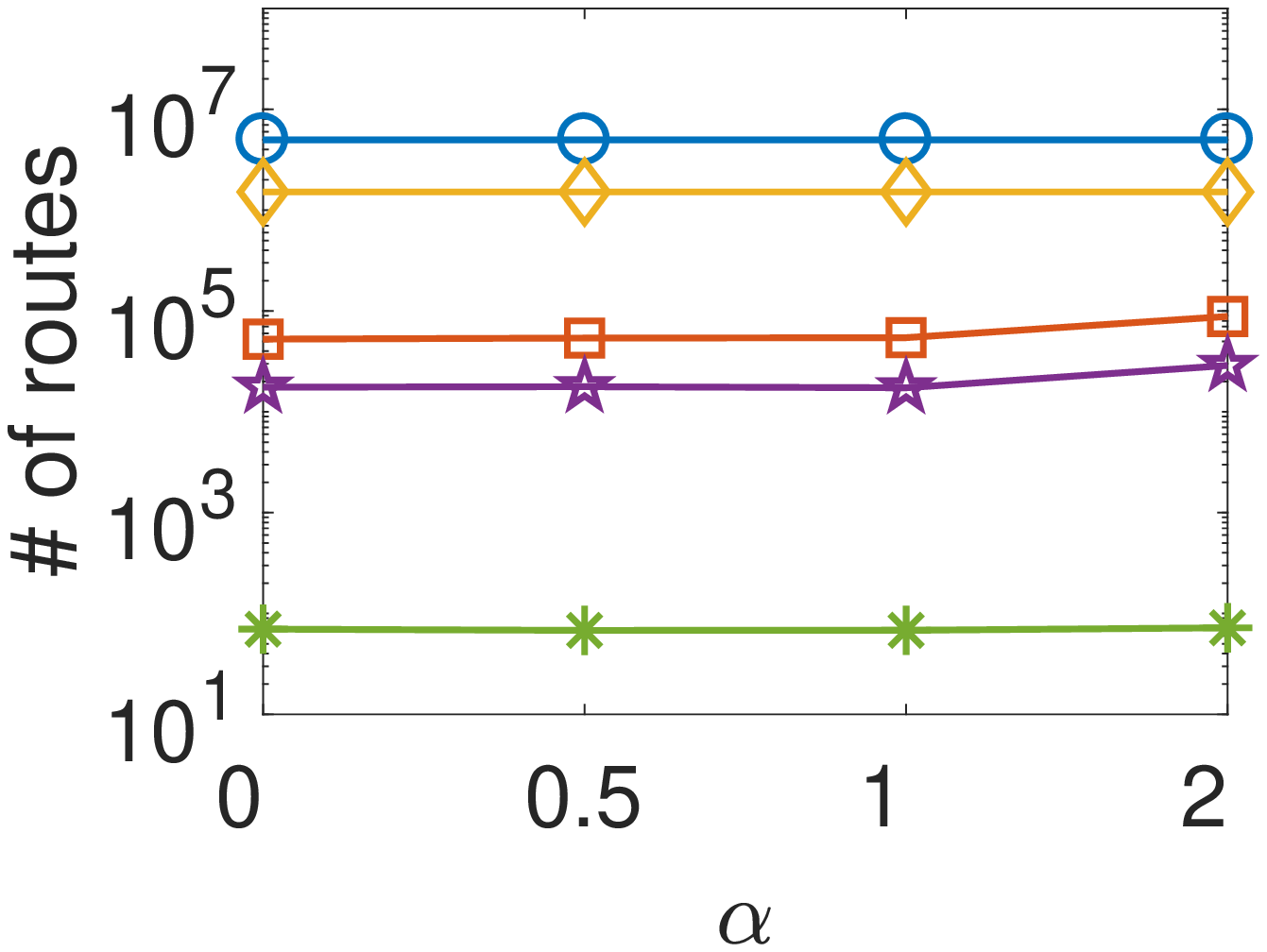}
                \caption{\# of routes vs. $\boldsymbol{\alpha}$}
                \label{fig:sg_alpha_s}
        \end{subfigure}
        ~ 
        \hspace{-0.15in}
        \begin{subfigure}[b]{0.165\textwidth}
                \includegraphics[width=\textwidth]{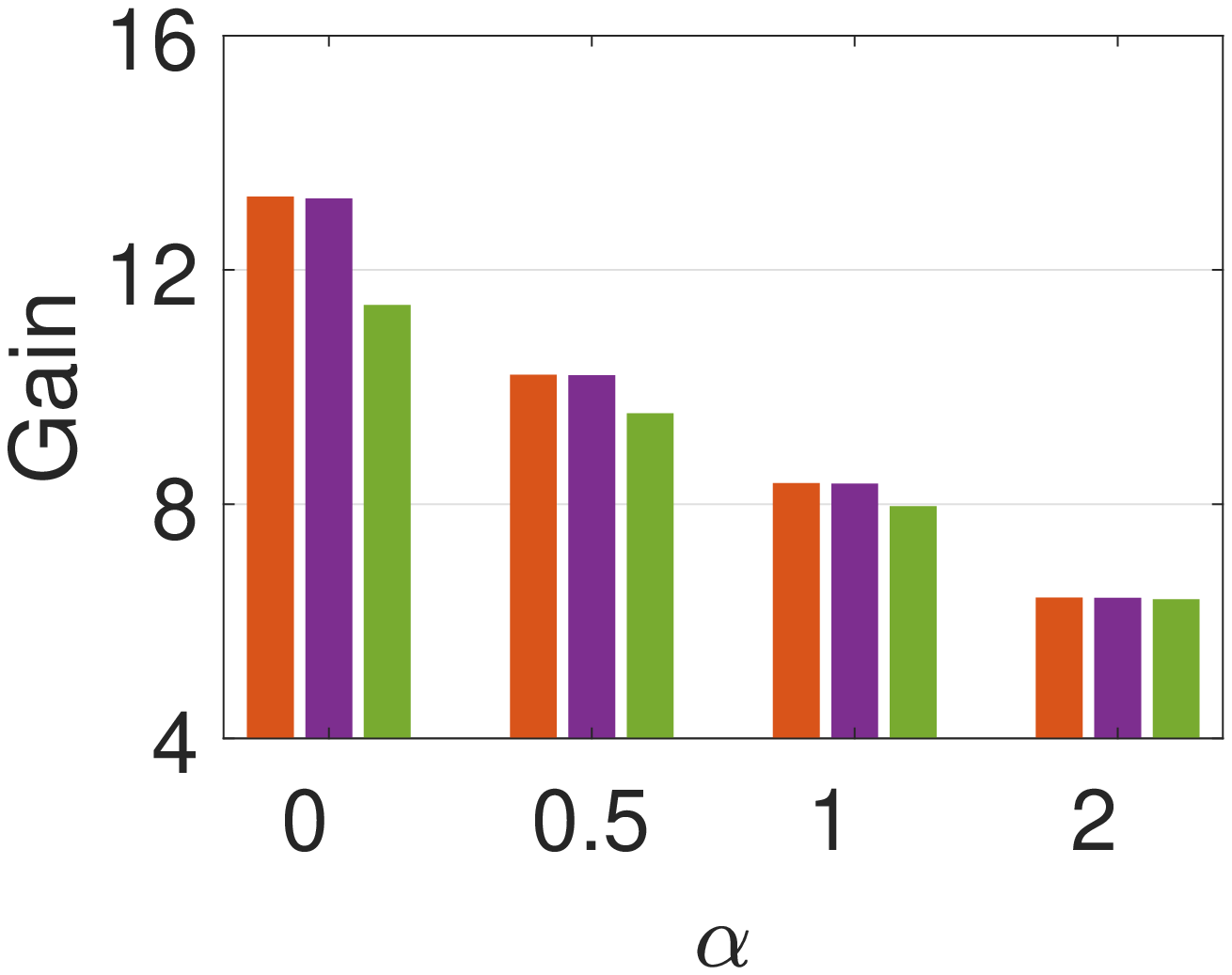}
                \caption{Gain vs. $\boldsymbol{\alpha}$}
                \label{fig:sg_alpha_g}
        \end{subfigure}

        \caption{Experimental results for \textsl{Singapore}. Run time and search space (\# of routes) are in logarithmic scale. The labels beside data points indicate the ratio of queries successfully responded by the algorithm under the parameter setting. No label if no query fail. Data point or bar is not drawn if more than half fail.
        AP can only respond queries with small $b$.
        }\label{fig:sg}
\end{figure}

\begin{figure}[t]  
        \hspace{0.2in}
        \begin{subfigure}[b]{0.25\textwidth}
                \includegraphics[width=\textwidth]{legend_1}
        \end{subfigure}%
        ~ 
        \hspace{0.25in}
        \begin{subfigure}[b]{0.15\textwidth}
                \includegraphics[width=\textwidth]{legend_2}
        \end{subfigure}%

        \begin{subfigure}[b]{0.165\textwidth}
                \includegraphics[width=\textwidth]{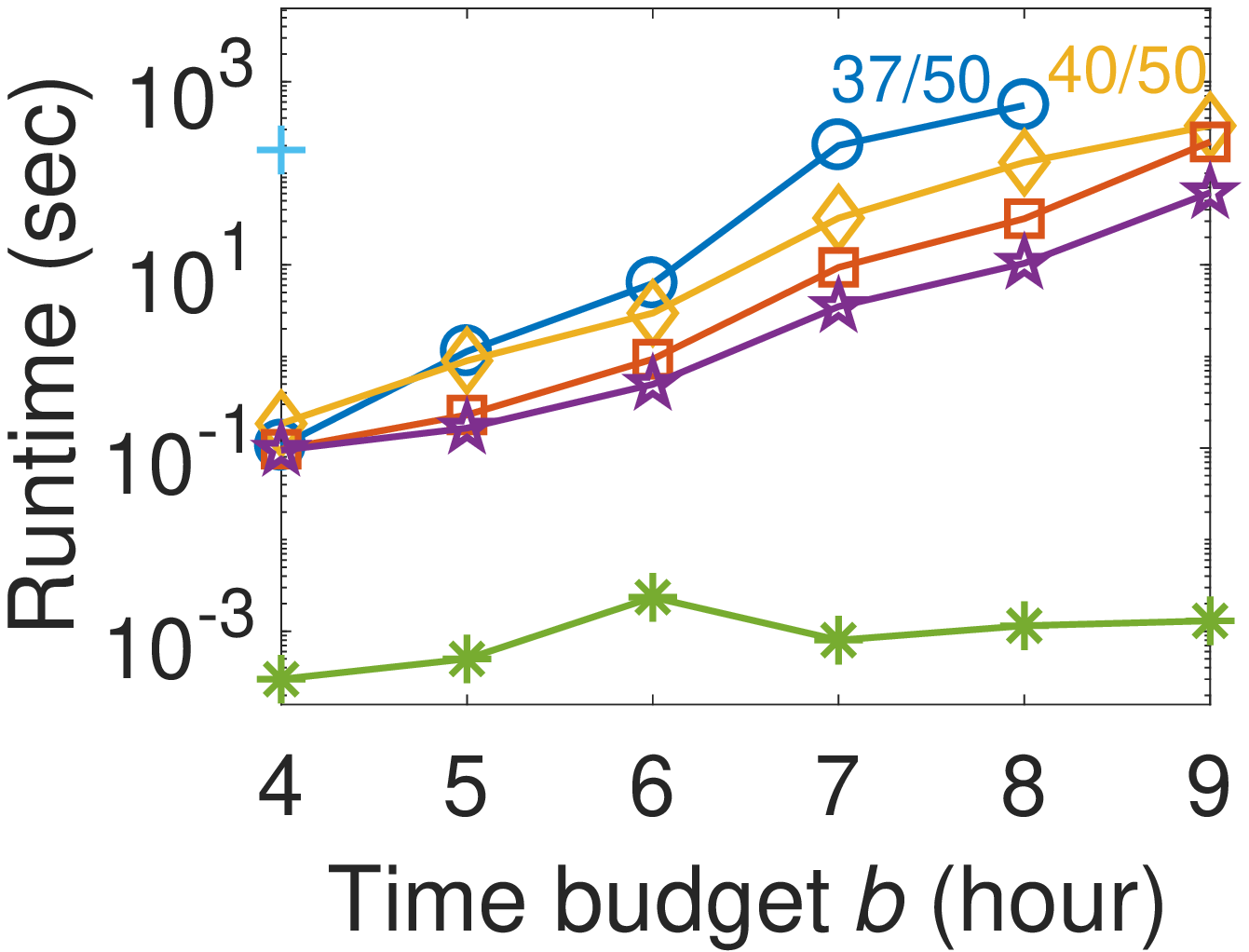}
                \caption{Runtime vs. $b$}
                \label{fig:as_b_t}
        \end{subfigure}%
        ~ 
        \hspace{-0.15in}
        \begin{subfigure}[b]{0.165\textwidth}
                \includegraphics[width=\textwidth]{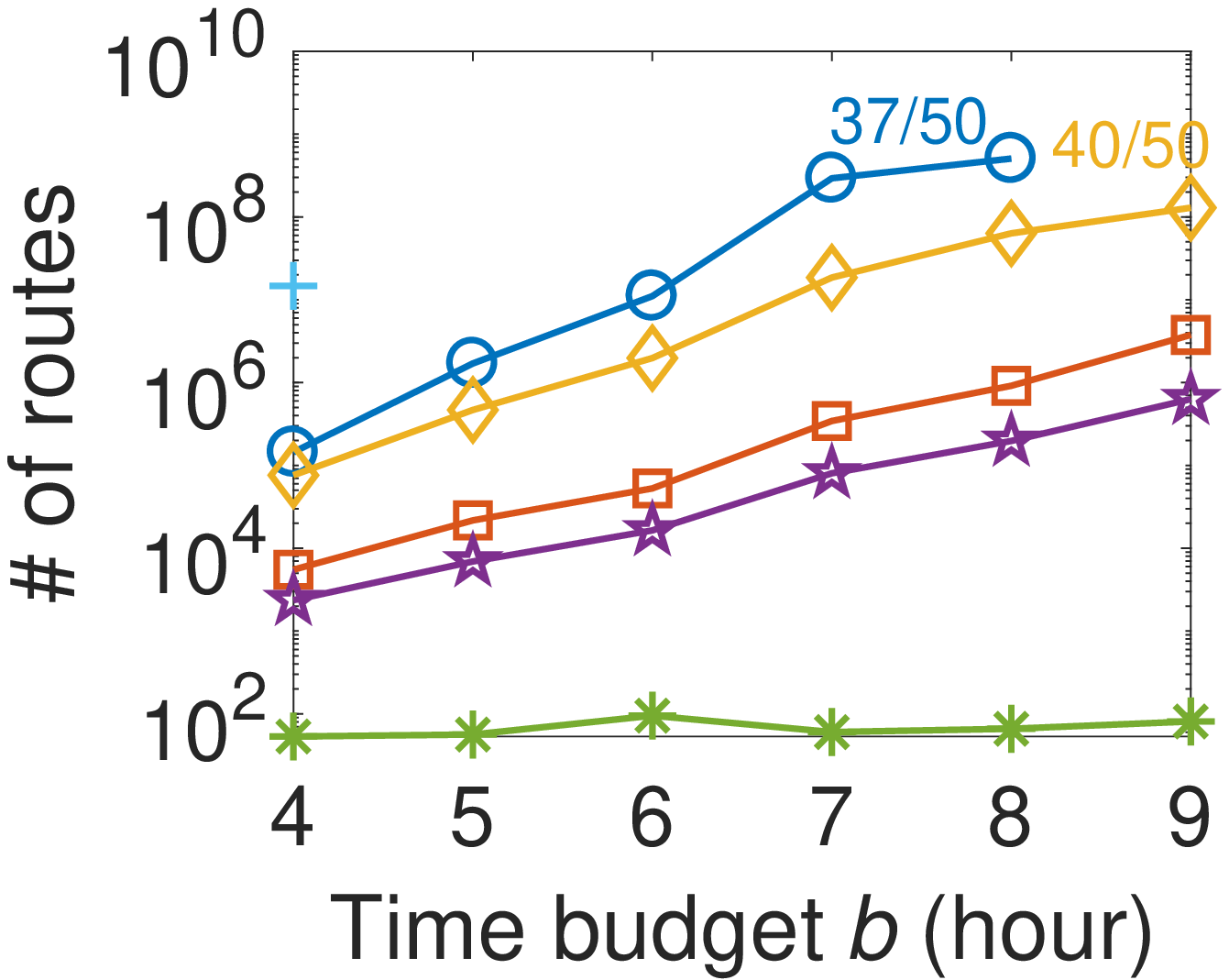}
                \caption{\# of routes vs. $b$}
                \label{fig:as_b_s}
        \end{subfigure}
        ~ 
        \hspace{-0.15in}
        \begin{subfigure}[b]{0.165\textwidth}
                \includegraphics[width=\textwidth]{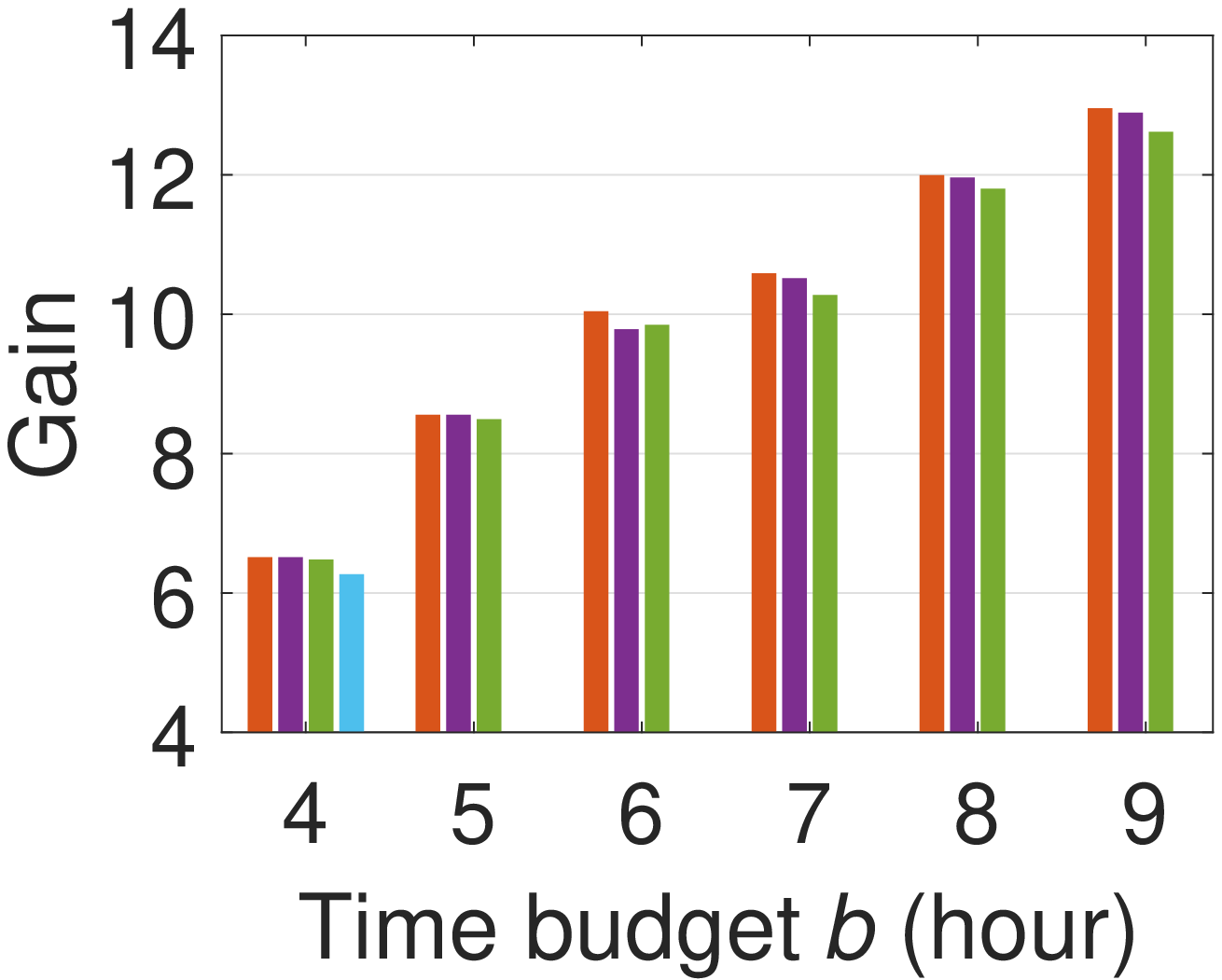}
                \caption{Gain vs. $b$}
                \label{fig:as_b_g}
        \end{subfigure}

        \begin{subfigure}[b]{0.165\textwidth}
                \includegraphics[width=\textwidth]{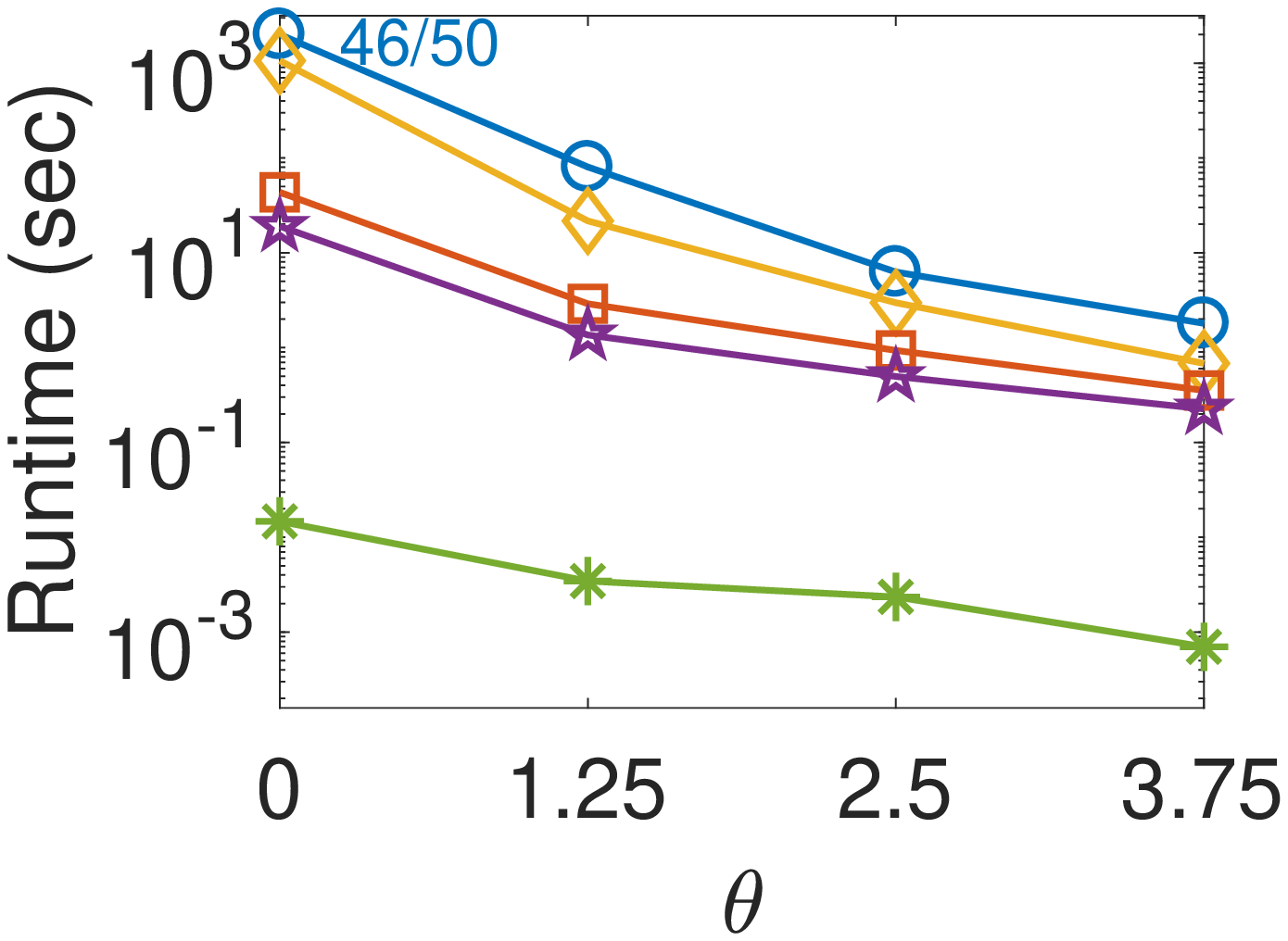}
                \caption{Runtime vs. $\boldsymbol{\theta}$}
                \label{fig:as_theta_t}
        \end{subfigure}%
        ~ 
        \hspace{-0.15in}
        \begin{subfigure}[b]{0.165\textwidth}
                \includegraphics[width=\textwidth]{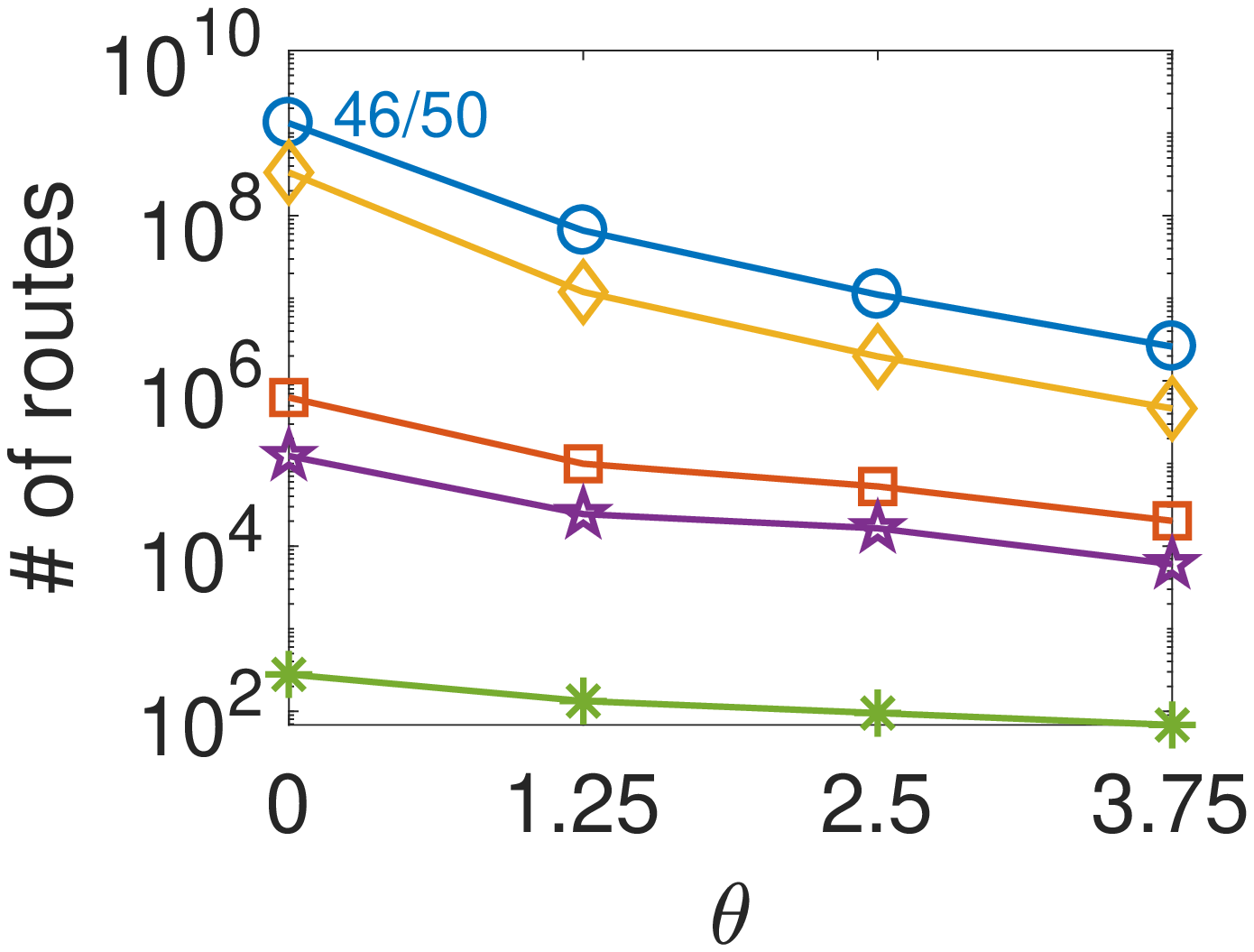}
                \caption{\# of routes vs. $\boldsymbol{\theta}$}
                \label{fig:as_theta_s}
        \end{subfigure}
        ~ 
        \hspace{-0.15in}
        \begin{subfigure}[b]{0.165\textwidth}
                \includegraphics[width=\textwidth]{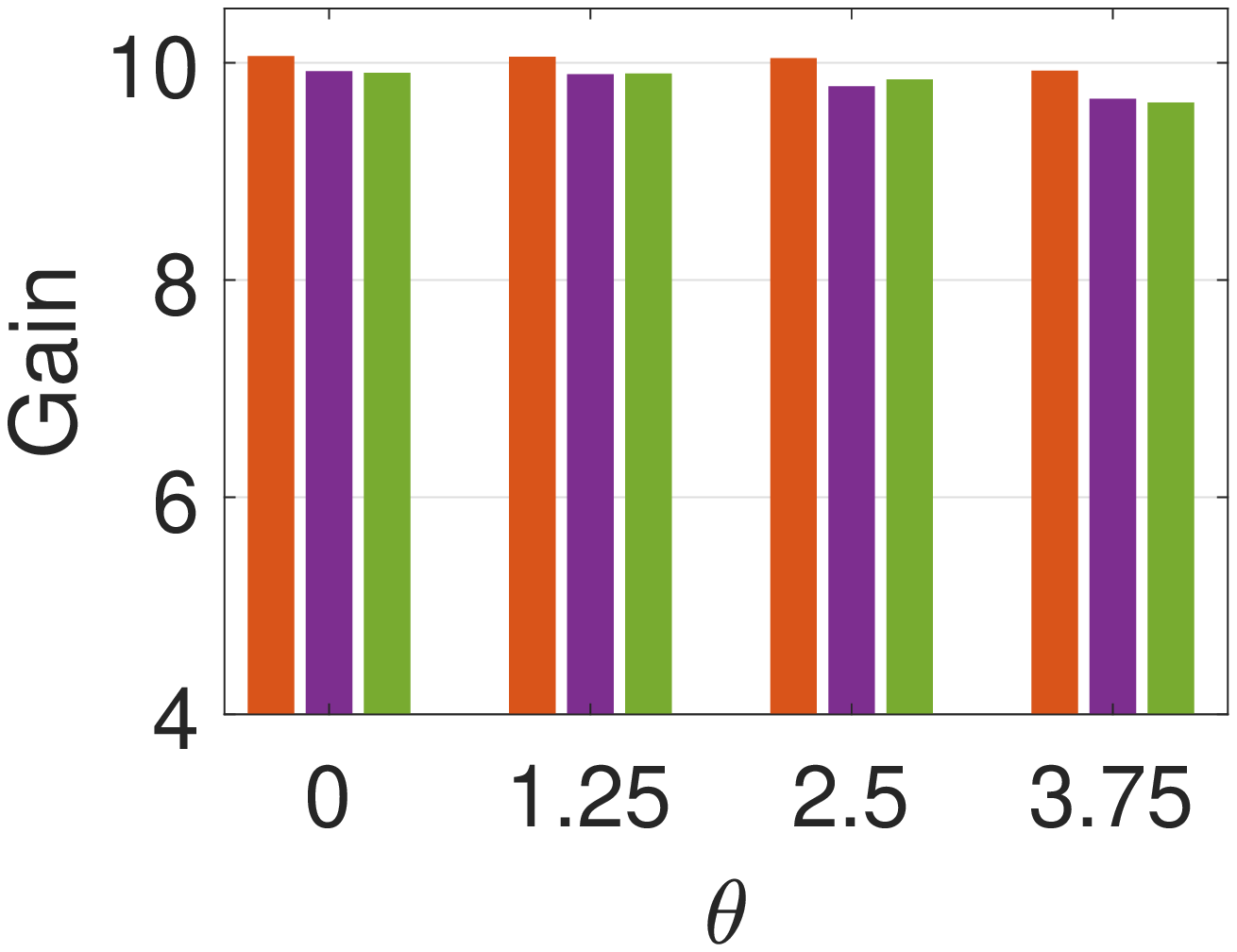}
                \caption{Gain vs. $\boldsymbol{\theta}$}
                \label{fig:as_theta_g}
        \end{subfigure}

        \begin{subfigure}[b]{0.165\textwidth}
                \includegraphics[width=\textwidth]{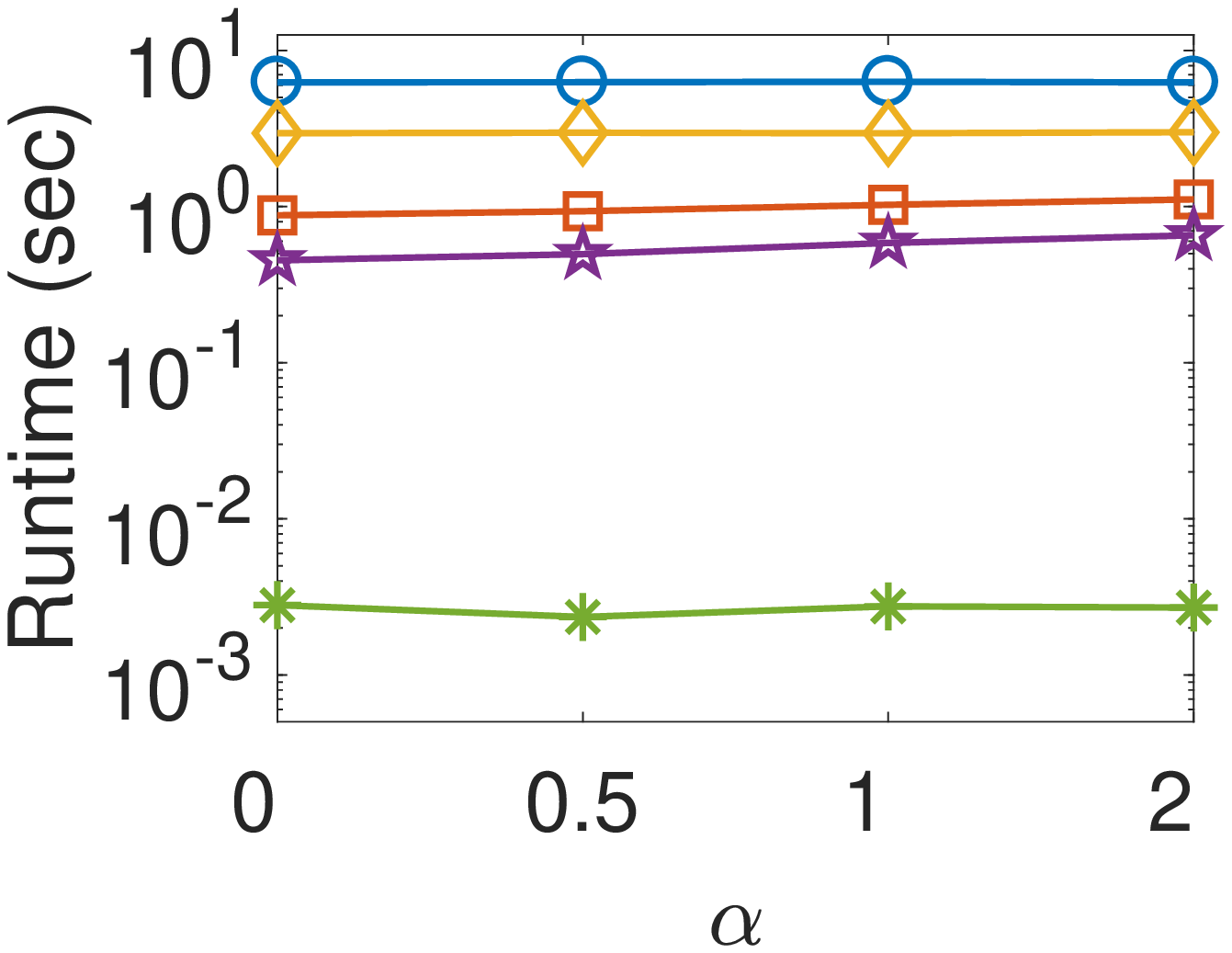}
                \caption{Runtime vs. $\boldsymbol{\alpha}$}
                \label{fig:as_alpha_t}
        \end{subfigure}%
        ~ 
        \hspace{-0.15in}
        \begin{subfigure}[b]{0.168\textwidth}
                \includegraphics[width=\textwidth]{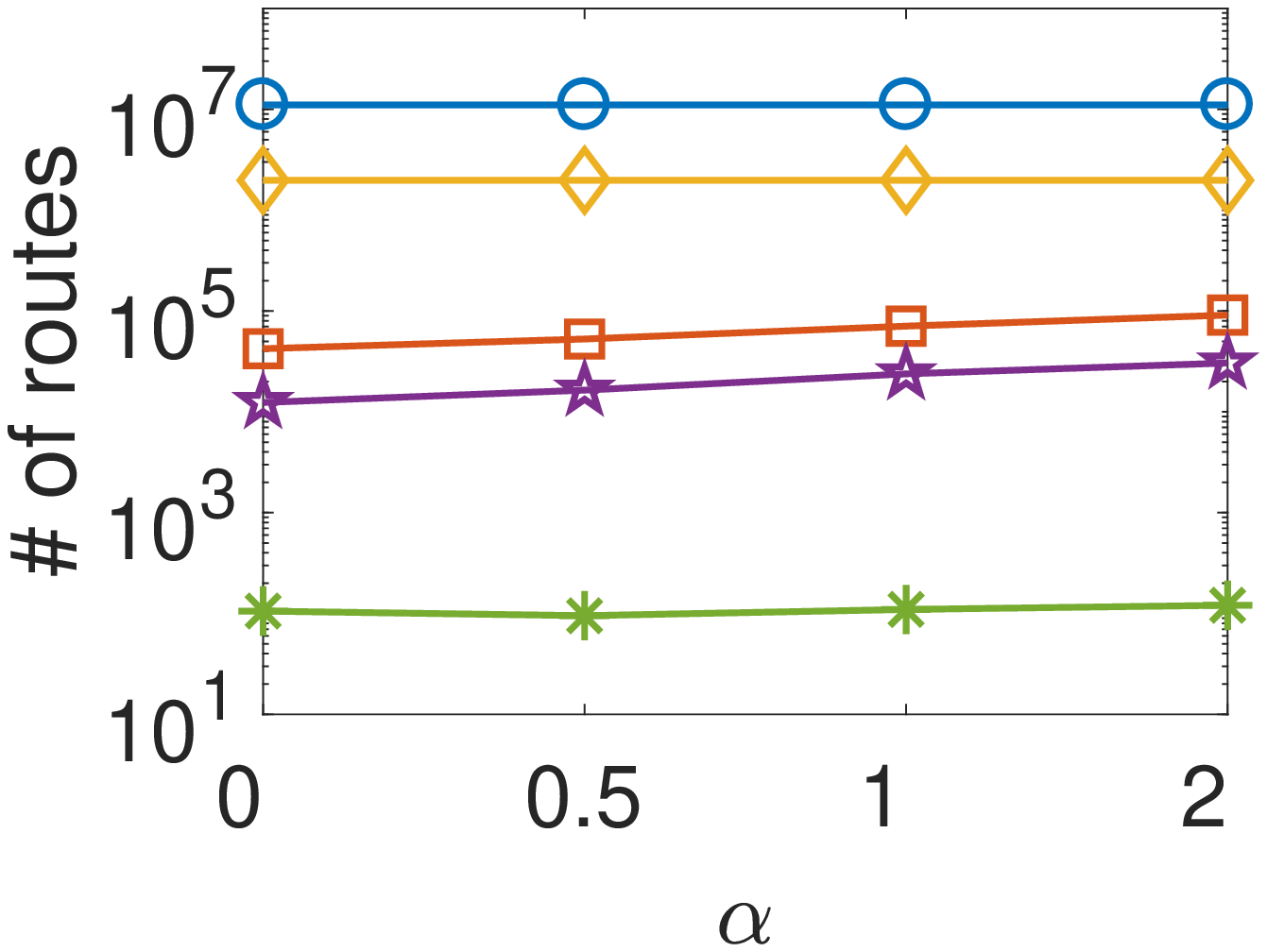}
                \caption{\# of routes vs. $\boldsymbol{\alpha}$}
                \label{fig:as_alpha_s}
        \end{subfigure}
        ~ 
        \hspace{-0.15in}
        \begin{subfigure}[b]{0.165\textwidth}
                \includegraphics[width=\textwidth]{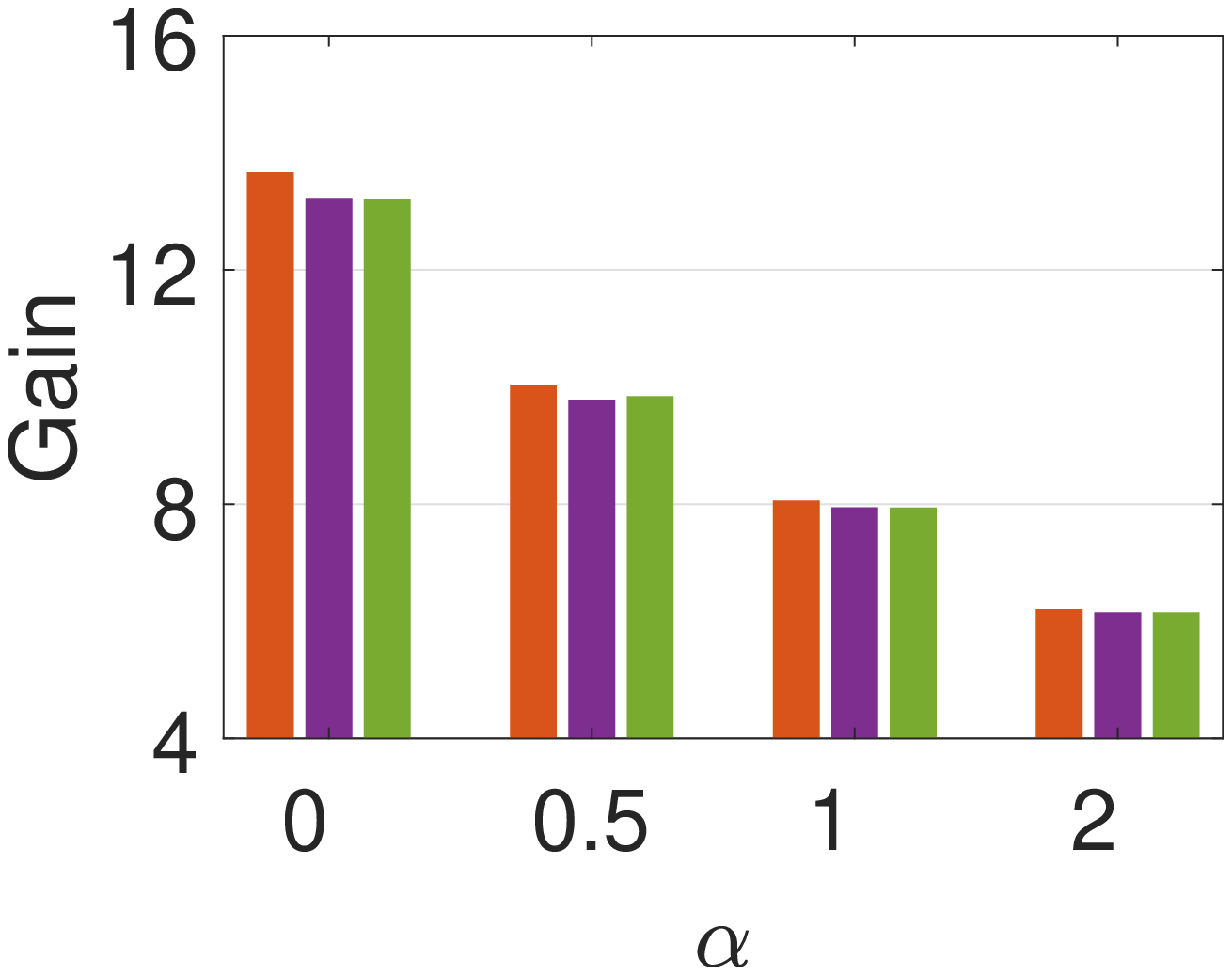}
                \caption{Gain vs. $\boldsymbol{\alpha}$}
                \label{fig:as_alpha_g}
        \end{subfigure}

        \caption{Experimental results for \textsl{Austin}}\label{fig:as}
\end{figure}

\subsubsection{Impact of budget $b$ (Figure \ref{fig:sg_b_t} - \ref{fig:sg_b_g} and \ref{fig:as_b_t} - \ref{fig:as_b_g})}  \label{sec:inf_b}

$b$ affects the length of routes (the number of POIs included). 


AP is the worst.
This is consistent with the analysis in Section \ref{sec:complex} that AP suffers from a high complexity when $b$ and $OPT$ have many discrete values. $b=6$ has 360 discrete values in minute, 
a majority of the queries cannot finish.
The efficiency of BF drops dramatically as $b$ increases,
since the number of open routes becomes huger and 
processing them is both time and memory consuming.

PACER+1's search space is two orders of magnitude smaller than that of BF, 
thanks to the compact state enumeration and the cost dominance pruning.
PACER+2 is the best among all the exact algorithms.
Compared with PACER+1, the one order of magnitude speedup in runtime and  
two orders of magnitude reduction in search space clearly demonstrates the additional pruning power of the Gain based upper bound pruning. 
PACER-SC trades optimality for efficiency.
Surprisingly, as shown in Figure \ref{fig:sg_b_g} and \ref{fig:as_b_g}, PACER-SC performs quite well with Gain being close to
that of OPTIMAL.

GR always finishes in less than $10^{-2}$ seconds.
For \textsl{Singapore}, the achieved gain is far worse than that of OPTIMAL, compared with the difference for \textsl{Austin}.
This is because $x$ and $y$ for \textsl{Singapore} are relatively remote to the central city. 
GR will greedily
select a POI $i$ not too far away from $x$ and $y$ (Eqn. (\ref{eq:gr_ratio})), thus, many POIs with possibly higher feature ratings located in the central city are less likely to be chosen. In contrast, $x$ and $y$ for \textsl{Austin} are in the downtown area and this situation is avoided in most cases.


%
%
%

\subsubsection{Impact of of filtering threshold $\boldsymbol{\theta}$}

In Figure \ref{fig:sg_theta_t} - \ref{fig:sg_theta_g} and \ref{fig:as_theta_t} - \ref{fig:as_theta_g}, 
as $\boldsymbol{\theta}$ increases, the POI candidate set becomes smaller and all the algorithms run faster. The majority of the queries for AP cannot finish and its results are not shown.
%
The study suggests that a reasonable value of $\boldsymbol{\theta}$, e.g., 2.5,
reduces the searching cost greatly while having little loss on
the quality of the found routes.

\subsubsection{Impact of route diversity parameter $\boldsymbol{\alpha}$ (Figure \ref{fig:sg_alpha_t} - \ref{fig:sg_alpha_g} and \ref{fig:as_alpha_t} - \ref{fig:as_alpha_g})}

PACER+2 and PACER-SC are slightly affected when $\boldsymbol{\alpha}$ varies.
As $\boldsymbol{\alpha}$ increases, the marginal return diminishes faster and
$\Phi_h$ behaves more towards the max aggregation. In this case,
Pruning-2 becomes less effective. When $\boldsymbol{\alpha} = 0$ (the sum aggregation), both $Gain$ and the difference between OPTIMAL and GR reach the maximum.

Figure \ref{fig:case} illustrates the effectiveness of our power law function in Eqn. (\ref{eq:gain_power_2}) for modeling the personalized route diversity requirement.
We run two
queries on \textsl{Singapore}, one with $\boldsymbol{\alpha}=0.5$, which specifies a diversity requirement, and one with $\boldsymbol{\alpha}=0$, which specifies the usual sum aggregation.
The other query parameters are the same.
The figures show the best routes found for each query, with the POIs on a route labeled sequentially as A, B $\cdots$. The red dots represent the source $x$ and destination $y$.
The route for $\boldsymbol{\alpha}=0.5$ covers all specified features, i.e., two POIs for each feature, while maximizing the total $Gain$.
While the route for $\boldsymbol{\alpha} = 0$ has four parks out of five POIs due to the higher weight of Park in $\mathbf{w}$;
thus, it is less preferred by a user who values diversity. 
In fact, the second route's $Gain$ value when evaluated using $\boldsymbol{\alpha}=0.5$ is only 6.60.

\begin{figure}[]  
        \begin{subfigure}[b]{0.42\textwidth}  
                \includegraphics[width=\textwidth]{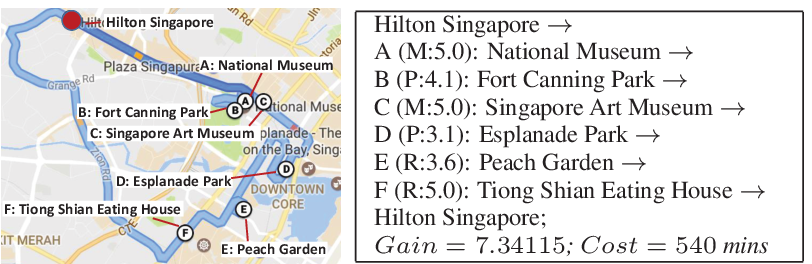}
                \caption{$\boldsymbol{\alpha} = 0.5$ (with diversity requirement)}
                \label{fig:alpha_0.5}
        \end{subfigure}%

        \begin{subfigure}[b]{0.42\textwidth}
                \includegraphics[width=\textwidth]{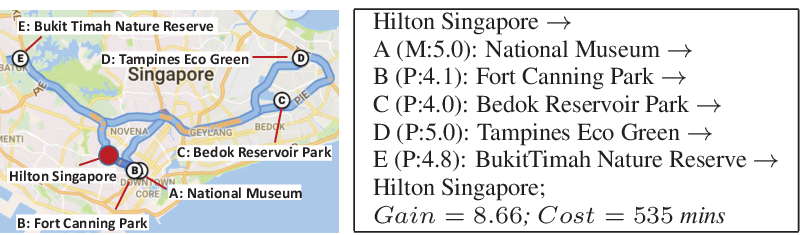}
                \caption{$\boldsymbol{\alpha} = 0$ (without diversity requirement)}
                \label{fig:alpha_0}
        \end{subfigure}%
		
        \caption{Two routes found from \textsl{Singapore} by PACER+2 for the query $Q=(x, y, b = 9, \mathbf{w} = (P:0.4, M:0.3, R:0.3), \boldsymbol{\theta} = 2.5, \boldsymbol{\alpha})$, where $x$ and $y$ are Hilton Singapore, and P, M and R represent Park, Museum, and Chinese Restaurant.} \label{fig:case}
\end{figure}

\subsubsection{Impact of $k$} \label{sec:impact_k}
We vary $k$ in range $[1,100]$ while fixing $b,\boldsymbol{\theta},\boldsymbol{\alpha}$ at the default values and run the algorithms, except GR and AP, on both datasets.
As $k$ only influences the gain-based pruning,
the performance of BF and PACER+1 are unchanged.
For PACER+2 and PACER-SC, the change is limited (less than 25\% slower for $k=100$). Because when $k$ is small, the Gain of the $k$-th best route is usually not far away to that of the best route, thus, the marginal gain upper bound pruning is not seriously influenced.
We omit the figures due to limited space.


\subsection{Comparison with A*} \label{sec:Astar}
A* \cite{zeng2015optimal} only works for their keyword coverage function:
$\Phi_{h}(\mathcal{P}_V) = 1- {\prod}_{i \in \mathcal{P}_V} [1 - \tilde{\mathbf{F}}_{i,h}]$,
and finds single route. In \cite{zeng2015optimal}, $\tilde{\mathbf{F}}_{i,h}$ is in the range $[0,1]$ and it is set to 1 if the number of check-ins on POI $i$ for feature $h$ is above average. In this case, the single POI in $\mathcal{P}$ yields the maximum $\Phi_{h}(\mathcal{P}_V)$ value; the feature $h$ of other POIs will be ignored.
For a fair comparison, we set $\beta = 0.5$ in Eqn. (\ref{eq:agg_rating}) for both algorithms, we also leverage our indices to speed up A*. 
Note that the maximum $b$ in \cite{zeng2015optimal} is 15 kilometers in their efficiency study, which is about 20 minutes by Google Maps under driving mode.

Figure \ref{fig:astar} shows the comparison between PACER+2 and the modified A* on both datasets.
The report of $Gain$ is omitted as they are both exact algorithms.
We also omit the comparison of search space due to page limit 
(PACER+2 searches one to two orders of magnitude less than A*).
Apparently, PACER+2 outperforms A*, especially for a large $b$. A few queries of A* on \textsl{Austin} even failed for $b = 9$.
Although A* has a pruning strategy specifically for their keyword coverage function, the search strategy itself is a bottleneck.
Besides, their pruning based on the greedy algorithm in \cite{khuller1999budgeted} has a bound looser than ours.
In fact, the experiments in \cite{zeng2015optimal} showed that A* is just 2-3 times faster than the brute-force algorithm.

\begin{figure}[]  
        \centering
        \begin{subfigure}[b]{0.17\textwidth}
                \includegraphics[width=\textwidth]{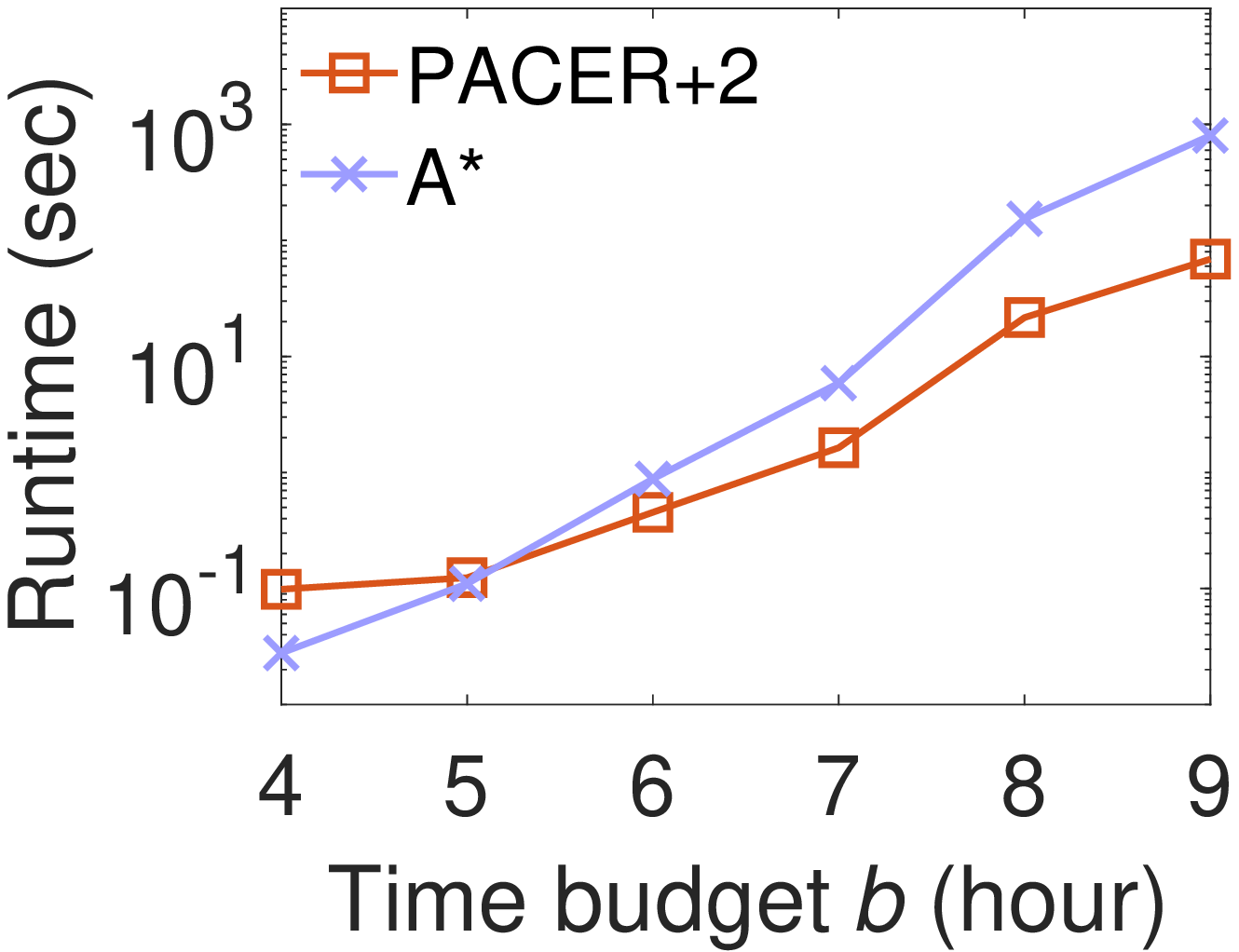}
                \caption{Runtime - \textsl{Singapore}}
                \label{fig:sg_cov_t}
        \end{subfigure}
        ~ 
        \hspace{0.0in}
%
        \begin{subfigure}[b]{0.17\textwidth}
        			\includegraphics[width=\textwidth]{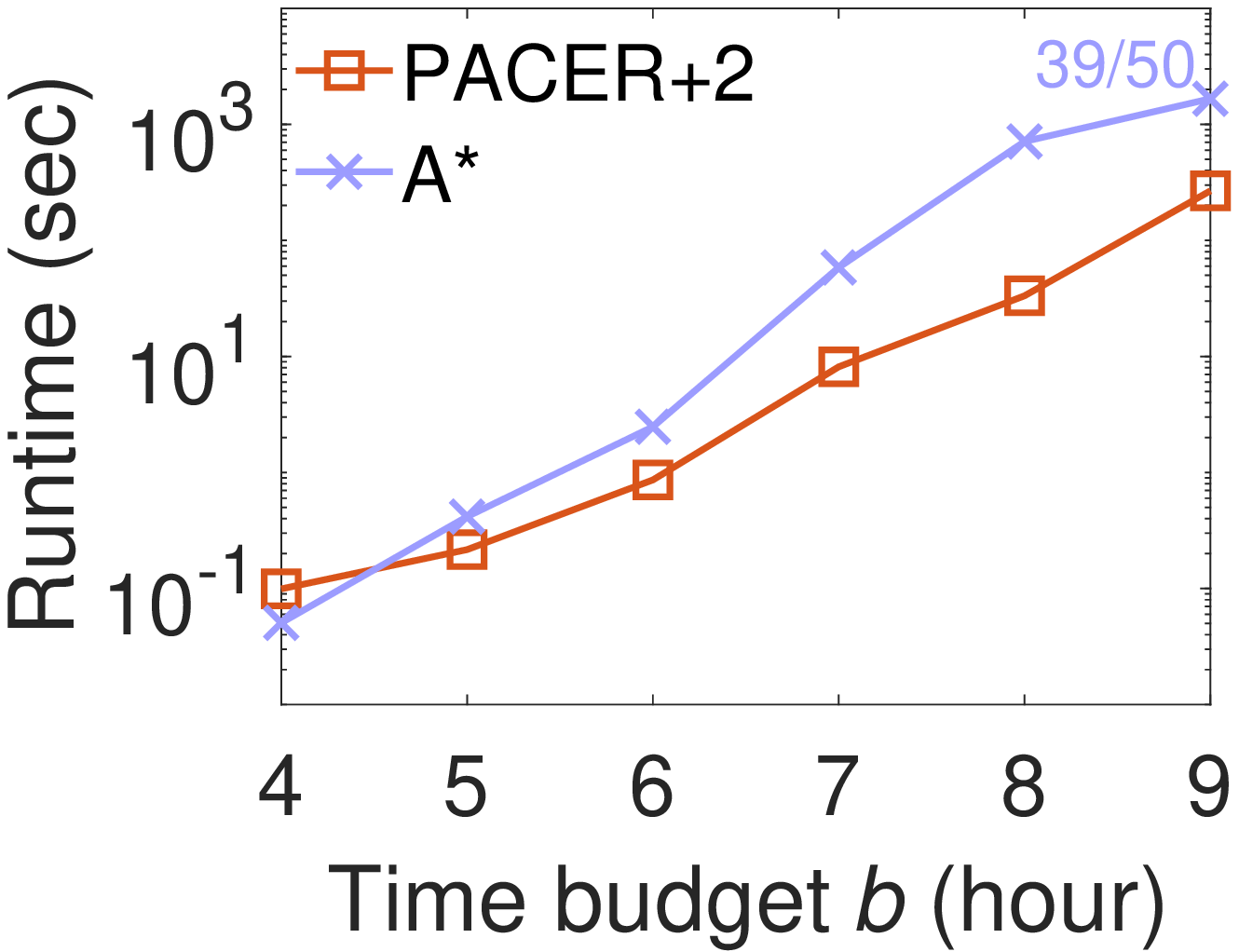}
            		\caption{Runtime - \textsl{Austin}}
                \label{fig:as_cov_t}
        \end{subfigure}
        \caption{PACER+2 vs. A* (logarithmic scale).}\label{fig:astar}
\end{figure}

\section{Conclusion} \label{sec:conclusion}
We considered a personalized top-$k$ route search problem. 
The large scale of POI maps and the combination of search in feature space and path space make this problem computationally hard. The personalized route diversity requirement further demands a solution that works for any reasonable route diversity specification. 
We presented an exact search algorithm with multiple pruning strategies to address these challenges.
We also presented high-performance heuristic solutions.
The analytical evaluation suggested 
that our solutions significantly outperform the state-of-the-art algorithms.

\textbf{Acknowledgments.} Ke Wang's work was partially supported by a discovery grant from The Natural Sciences and Engineering Research Council of Canada (NSERC).


\balance


\bibliographystyle{ACM-Reference-Format}
\bibliography{ref}

\end{document}